\documentclass[11pt,a4paper]{article}

\usepackage[margin=1in]{geometry}
\usepackage[T1]{fontenc}
\usepackage{lmodern}
\usepackage{microtype}
\usepackage{amsmath,amssymb,amsthm}
\usepackage{mathtools}
\usepackage[table]{xcolor}
\usepackage{tikz}
\usetikzlibrary{arrows.meta}
\usepackage[hidelinks]{hyperref}
\usepackage[nameinlink,capitalise]{cleveref}
\usepackage{aliascnt}

\newtheorem{theorem}{Theorem}[section]

\newaliascnt{lemma}{theorem}
\newtheorem{lemma}[lemma]{Lemma}
\aliascntresetthe{lemma}
\crefname{lemma}{Lemma}{Lemmas}

\newaliascnt{proposition}{theorem}
\newtheorem{proposition}[proposition]{Proposition}
\aliascntresetthe{proposition}
\crefname{proposition}{Proposition}{Propositions}

\newaliascnt{corollary}{theorem}
\newtheorem{corollary}[corollary]{Corollary}
\aliascntresetthe{corollary}
\crefname{corollary}{Corollary}{Corollaries}

\theoremstyle{definition}

\newaliascnt{definition}{theorem}
\newtheorem{definition}[definition]{Definition}
\aliascntresetthe{definition}
\crefname{definition}{Definition}{Definitions}

\newaliascnt{example}{theorem}
\newtheorem{example}[example]{Example}
\aliascntresetthe{example}
\crefname{example}{Example}{Examples}

\theoremstyle{remark}

\newaliascnt{remark}{theorem}
\newtheorem{remark}[remark]{Remark}
\aliascntresetthe{remark}
\crefname{remark}{Remark}{Remarks}

\title{Separating Parsing Expression Grammars using Cell-Probe Lower Bounds}
\author{
Jungyeom Kim\\
\small Independent Researcher\\
\small\texttt{kimjg1199lock@gmail.com}
\and
Jihyeok Park\\
\small Korea University\\
\small\texttt{jihyeok\_park@korea.ac.kr}
}

\begin{document}

\newcommand{\nonterminalset}{\mathcal{N}}
\newcommand{\nonterminal}{A}

\newcommand{\symset}{\Sigma}
\newcommand{\sym}{a}

\newcommand{\productionset}{P}
\newcommand{\cfgderive}[1]{\mathrel{\Rightarrow_{#1}}}

\newcommand{\pegrelation}{R}
\newcommand{\pegstart}{S}
\newcommand{\parsingexpressions}{\mathcal{E}}
\newcommand{\pegfail}{\mathsf{fail}}
\newcommand{\pegsem}[4]{#1 \vdash #2,#3 \Downarrow #4}
\newcommand{\pegrule}[1]{\;\text{\scriptsize(\textsc{#1})}}

\newcommand{\automaton}{A}
\newcommand{\stateset}{Q}
\newcommand{\workalphabet}{\Gamma}
\newcommand{\scaffold}{\mathcal{S}}
\newcommand{\scamissing}{\bot}
\newcommand{\scaself}{\mathsf{SELF}}
\newcommand{\scapaths}[2]{[#1]^{\le #2}}
\newcommand{\scaview}[3]{\mathcal{B}_{#1}(#2,#3)}
\newcommand{\scastep}{\operatorname{Step}}

\newcommand{\bitset}{\{0,1\}}
\newcommand{\celladdr}[1]{\mathsf{Addr}_{#1}}
\newcommand{\cellword}[1]{\mathsf{Word}_{#1}}
\newcommand{\cellmemory}[1]{\mathsf{Mem}_{#1}}

\maketitle

\begin{abstract}
We resolve three open problems concerning parsing expression grammars (PEGs).
We construct a single language $C$ satisfying
$C\in\mathsf{LIN}\cap\mathsf{PEG}$ and
$C^R\in\mathsf{LIN}\setminus\mathsf{PEG}$.
This proves that some linear context-free language is not a PEG language and
that PEG languages are not closed under reversal, confirming a conjecture of
Loff, Moreira, and Reis.
Factoring the same witness resolves the concatenation-closure problem of
Rubtsov and Chudinov negatively, in the strong form
$\mathsf{PEG}\cdot\mathsf{REG}\not\subseteq\mathsf{PEG}$ despite
$\mathsf{REG}\cdot\mathsf{PEG}\subseteq\mathsf{PEG}$.
It also refutes closure under Kleene star, homomorphisms, and substitutions.

Our main technique converts scaffolding automata (SCAs), which characterize
reversals of PEG languages, into dynamic data structures in the cell-probe
model.
For any suitably local serialization of a problem with preprocessing, updates,
and a final Boolean query, an SCA recognizer yields an exact deterministic
cell-probe data structure whose operation costs are proportional to the
corresponding encoding lengths.
Cell-probe lower bounds can therefore prove SCA non-membership and, by reversal,
PEG non-membership.
We apply this transfer to Multiphase Inner Product using one-symbol update
blocks and a query suffix of length $\mathcal{O}(\log n)$, while keeping both
the language and its reversal linear context-free.
Ko's cell-probe lower bound then yields the witness above.
The arguments are additionally formalized in Lean~4.
\end{abstract}

\section{Introduction}

Parsing expression grammars (PEGs) are formal grammars that define languages through parsing rules, rather than through generative rules as in context-free grammars (CFGs)~\cite{Ford2004PEG}.
PEGs have several notable advantages. In particular, their prioritized choice gives them deterministic recognition semantics
and makes them unambiguous by construction~\cite{Ford2004PEG}.
Moreover, PEGs can be parsed in linear time using a packrat parser, at the cost of additional memory~\cite{Ford2002Packrat}.
Packrat parsing is essentially a memoization technique based on dynamic programming.
Owing to these advantages, PEGs have been adopted in practical parser implementations.
CPython replaced its LL(1) parser with a PEG-based parser in Python 3.9~\cite{PEP617}.

Although PEGs are becoming increasingly influential in parsing, their theoretical properties are not yet fully understood.
PEGs are highly expressive, and every deterministic context-free language can be expressed by a
PEG~\cite{BirmanUllman1973,Ford2004PEG}.
Furthermore, PEGs can express some non-context-free languages.
One such example is $\{a^n b^n c^n \mid n \geq 0\}$~\cite{BirmanUllman1973,Ford2004PEG}.
Therefore, writing $\mathsf{PEG}$, $\mathsf{DCFL}$, and $\mathsf{CFL}$ for the classes of languages defined by PEGs,
deterministic context-free grammars, and context-free grammars, respectively, both
\(
\mathsf{DCFL} \subseteq \mathsf{PEG}
\)
and
\(
\mathsf{PEG} \not \subseteq \mathsf{CFL}
\)
hold.
PEGs are, moreover, closed under complement and intersection, owing to the expressive power of their lookahead
operators~\cite{Ford2004PEG}.
Despite this expressive power, the following two questions have remained open:
\begin{itemize}
\item[Q1.] Does $\mathsf{CFL} \subseteq \mathsf{PEG}$ hold?
\item[Q2.] Is the class of PEG languages closed under reversal, that is, does $\mathsf{PEG}=\mathsf{PEG}^R$ hold?
\end{itemize}

We prove that there exists a linear context-free language that cannot be recognized by any PEG and that the class of PEG languages is not closed under reversal.

Both results require showing that no PEG recognizes a given language, which is particularly challenging.
Loff, Moreira, and Reis asked whether the language of palindromes has a PEG, observing that none was known
even for that comparatively simple language~\cite{LoffMoreiraReis2020}; our witness is instead a language
built for the reduction below.
Other language classes have well-established techniques for proving that a given language does not belong to them.
For example, regular languages can be characterized using the Myhill--Nerode theorem, while context-free languages can often be ruled out using the pumping lemma~\cite{Myhill1957,Nerode1958,BarHillelPerlesShamir1961}.
By contrast, no comparably effective tool is known for PEGs.
Indeed, Loff, Moreira, and Reis proved that PEGs do not admit a pumping lemma analogous to the one for context-free languages~\cite{LoffMoreiraReis2020}.

Known positive results do not close this gap from the other side either: translations into PEGs exist only for
restricted grammar classes.
Mascarenhas, Medeiros, and Ierusalimschy gave language-preserving translations from several subclasses of CFGs,
including LL-regular grammars, to PEGs~\cite{MascarenhasMedeirosIerusalimschy2014}.
Our first result shows that such translations cannot cover all context-free languages.

The difficulty of proving PEG non-membership has motivated several automata-theoretic models for studying the
expressive power of PEGs.
Loff, Moreira, and Reis introduced \emph{scaffolding automata} (SCAs), a finite-automaton model enriched with a
pointer structure, and conjectured that the answer to Q2 is negative~\cite{LoffMoreiraReis2020}.
Rubtsov and Chudinov later introduced \emph{deterministic pointer pushdown automata} as another such model:
pushdown automata that record with each stack symbol the input position at which it was pushed and may return the
input head there when the symbol is popped~\cite{RubtsovChudinov2024}.
Among these models, our proof builds on scaffolding automata, which we now describe in more detail.

An SCA extends a finite automaton with an append-only persistent pointer structure called a scaffold.
As the automaton reads each input symbol, it appends a new node carrying a label and a constant number of backward or self pointers.
The scaffold thus grows only by edges pointing toward already-created nodes, encoding the input history in
chronological order; this inherent directionality provides important intuition for why reversal may substantially
alter the expressive power of PEGs.
The key relationship between PEGs and SCAs is captured by the following characterization.
For every language $H$,
\[
H \in \mathsf{PEG}
\iff
H^{R} \in \mathsf{SCA},
\]
where $\mathsf{SCA}$ denotes the class of languages decided by SCAs.
This characterization reduces PEG lower bounds to SCA lower bounds, but it does not by itself supply such bounds.

Although the scaffold retains the entire input history, an SCA can access only a constant-size portion of it while
processing each input symbol.
We exploit this access bottleneck through a framework that interprets an SCA as an online data structure.
We call a Boolean data-structure problem \emph{three-stage} if, for each size parameter $n$, it consists of preprocessing data, a sequence of updates, and a final query with a Boolean answer.
We serialize these stages as a prefix, a sequence of update blocks, and a query suffix, so that membership of the
resulting string in a fixed language $K$ determines the Boolean answer.
Each block depends only on the operation it encodes (an update block may also depend on its position in the
sequence), and every complete serialization has polynomial length.
If, in addition, the update blocks and the query suffix have lengths at most $a(n)$ and $b(n)$, respectively, we
call $K$ an $(a,b)$-$\mathsf{SCA}$ encoding (\cref{def:sca-encoding}).

Our transfer theorem (\cref{thm:transfer}) states that if such an encoding is decided by an SCA, then the
underlying problem admits an exact deterministic cell-probe data structure with word size $\Theta(\log n)$, update
time $\mathcal{O}(a(n))$, query time $\mathcal{O}(b(n))$, and polynomially many memory cells.
Contrapositively, excluding such a data structure establishes that the encoding lies outside $\mathsf{SCA}$, and
the PEG--SCA characterization then places its reversal outside $\mathsf{PEG}$.

We apply our strategy to Multiphase Inner Product, a variant of P\u{a}tra\c{s}cu's Multiphase
Problem~\cite{Patrascu2010Multiphase} for which Ko proved the sufficiently strong lower bound used
here~\cite{Ko2026CellProbe}.
Our encoding uses constant-length update blocks and a query suffix of length $\mathcal{O}(\log n)$.
Consequently, any query-time lower bound of $\omega(\log n)$, that is, any bound exceeding the logarithmic
barrier, suffices for our purposes.
The resulting witness language and its reversal are both linear context-free; we write $\mathsf{LIN}$ for the
class of linear context-free languages.

The same witness also resolves further closure questions.
We prove that PEG languages are closed under neither unrestricted concatenation nor Kleene star, with concatenation
failing even in the strong form
\[
\mathsf{PEG}\cdot\mathsf{REG}\nsubseteq\mathsf{PEG}.
\]

Our main contributions are as follows:
\begin{enumerate}
\item We introduce a framework that transforms an SCA encoding of a dynamic data-structure problem into a concrete exact deterministic cell-probe data structure with bounded word size, update time, and query time.
\item Using this framework, we construct a language $C$ satisfying
\[
C\in\mathsf{LIN}\cap\mathsf{PEG}
\qquad\text{and}\qquad
C^R\in\mathsf{LIN}\setminus\mathsf{PEG}.
\]
Consequently, $\mathsf{LIN}\not\subseteq\mathsf{PEG}$ and hence
$\mathsf{CFL}\not\subseteq\mathsf{PEG}$, while the class of PEG languages is not closed under reversal.
\item From the same witness, we resolve the open closure questions for concatenation and Kleene star negatively.
As further consequences, PEG languages are not closed under homomorphisms or substitutions.
\end{enumerate}

\noindent
Beyond these main contributions, we provide a Lean~4 formalization of the separation argument for the witness
language, whose precise scope and trust assumptions are described in the
\hyperref[sec:formal-verification]{Formal Verification} section.

\section{Preliminaries}
\label{sec:preliminaries}

\subsection{Notation}

An alphabet is a finite set of symbols.
For an alphabet $\symset$, we write $\symset^*$ for the set of finite words over $\symset$, $\varepsilon$ for the empty
word, and $\symset^+=\symset^*\setminus\{\varepsilon\}$.
The length and reversal of a word $x$ are denoted by $|x|$ and $x^R$, respectively.
For a language $L\subseteq\symset^*$, let $L^R=\{\,x^R:x\in L\,\}$.
Likewise, for a class $\mathcal{C}$ of languages, let
$\mathcal{C}^R=\{\,L^R:L\in\mathcal{C}\,\}$; the class $\mathcal{C}$ is \emph{closed under reversal} if
$\mathcal{C}=\mathcal{C}^R$.
For an integer $m\geq0$, let $[m]=\{0,\ldots,m-1\}$.

\subsection{Context-Free Grammars and Linear Languages}

\begin{definition}[context-free grammar]
A context-free grammar is a tuple \(G=(\nonterminalset,\symset,\productionset,\pegstart)\).
Here $\nonterminalset$ is a finite set of nonterminal symbols, $\symset$ is a finite set of terminal symbols disjoint
from $\nonterminalset$, and $\pegstart\in\nonterminalset$ is the start symbol.
The set $\productionset$ is a finite set of productions contained in
$\nonterminalset\times(\nonterminalset\cup\symset)^*$.
Each pair $(\nonterminal,\alpha)\in\productionset$ represents the production
$\nonterminal\to\alpha$.
For $x,y\in(\nonterminalset\cup\symset)^*$, define the one-step derivation relation by
\[
x\cfgderive{G}y
\quad\Longleftrightarrow\quad
\begin{aligned}
\exists\,&u,v\in(\nonterminalset\cup\symset)^*,\
\nonterminal\in\nonterminalset,\
\alpha\in(\nonterminalset\cup\symset)^*
\\
&x=u\nonterminal v,\qquad
y=u\alpha v,\qquad
(\nonterminal,\alpha)\in\productionset.
\end{aligned}
\]
The language generated by $G$ is
\[
L(G)=
\{\,w\in\symset^*:\pegstart\cfgderive{G}^{*}w\,\}.
\]
\end{definition}

\begin{definition}[linear context-free language]
A context-free grammar is \emph{linear} if the right-hand side of every production contains at most one nonterminal.
Every sentential form reachable from the start symbol therefore contains at most one nonterminal, so the recursive
structure of a derivation is linear.
We write
\[
\begin{aligned}
\mathsf{CFL}
  &=\{\,L(G):G\text{ is a context-free grammar}\,\},\\
\mathsf{LIN}
  &=\{\,L(G):G\text{ is a linear context-free grammar}\,\}.
\end{aligned}
\]
\end{definition}

\subsection{Parsing Expression Grammars}

Unlike a context-free grammar, a parsing expression grammar specifies a deterministic recognition procedure.
Its choice operator is prioritized, and its lookahead operators may inspect the remaining input without consuming it.

\begin{definition}[parsing expression]
For disjoint finite sets of nonterminals $\nonterminalset$ and terminals $\symset$, the set
$\parsingexpressions(\nonterminalset,\symset)$ of parsing expressions is generated by
\[
e ::= \varepsilon \mid \pegfail \mid \sym \mid \nonterminal
   \mid e_1e_2 \mid e_1/e_2 \mid {!e},
\qquad
\sym\in\symset,\quad \nonterminal\in\nonterminalset.
\]
\end{definition}

The prioritized choice $e_1/e_2$ first attempts $e_1$ and attempts $e_2$ only if $e_1$ fails.
Lookahead evaluates an expression while retaining only whether it succeeds, rather than the input position reached by
the evaluation.
The result of an entire recursive parse can therefore guide subsequent parsing from the original input position.
Negative lookahead $!e$ succeeds if and only if $e$ fails, while positive lookahead is given by the syntactic sugar
$\mathbin{\&}e\coloneqq !(!e)$ and succeeds if and only if $e$ succeeds.
We also use a dot for the choice of all terminals, writing \(.=\sym_1/\cdots/\sym_m\) when
\(\symset=\{\sym_1,\ldots,\sym_m\}\), and hence $!.$ tests for the end of the input.
When an expression $e$ is guaranteed to consume at least one symbol whenever it succeeds, $e^*$ denotes the usual
abbreviation obtained by introducing a fresh nonterminal $A_e\leftarrow eA_e/\varepsilon$.

\begin{definition}[parsing expression grammar]
A parsing expression grammar (PEG) is a tuple
\(G=(\nonterminalset,\symset,\pegrelation,\pegstart)\).
Here $\nonterminalset$ is a finite set of nonterminals, $\symset$ is a finite terminal alphabet disjoint from
$\nonterminalset$, and $\pegstart\in\nonterminalset$ is the start nonterminal.
The rule function
\(\pegrelation\colon\nonterminalset\to\parsingexpressions(\nonterminalset,\symset)\) assigns a parsing expression
$\pegrelation(\nonterminal)$ to each nonterminal $\nonterminal\in\nonterminalset$.
We write $\nonterminal\leftarrow e$ when $\pegrelation(\nonterminal)=e$.
\end{definition}

For $r\in\symset^*\cup\{\pegfail\}$, the judgment $\pegsem{G}{e}{w}{r}$ means that $e$ leaves the suffix $r$ of $w$
when $r\ne\pegfail$; the result $\pegfail$ records failure.
The judgment is defined by the following rules.
\[
\begin{gathered}
\frac{}{\pegsem{G}{\varepsilon}{w}{w}}
\pegrule{Eps}
\qquad
\frac{}{\pegsem{G}{\pegfail}{w}{\pegfail}}
\pegrule{Fail}
\\[1ex]
\frac{}{\pegsem{G}{\sym}{\sym w}{w}}
\pegrule{Term-S}
\qquad
\frac{w\notin \sym\symset^*}
     {\pegsem{G}{\sym}{w}{\pegfail}}
\pegrule{Term-F}
\\[1ex]
\frac{\pegsem{G}{\pegrelation(\nonterminal)}{w}{r}}
     {\pegsem{G}{\nonterminal}{w}{r}}
\pegrule{Var}
\\[1ex]
\frac{\pegsem{G}{e_1}{w}{y}\quad
      \pegsem{G}{e_2}{y}{r}}
     {\pegsem{G}{e_1e_2}{w}{r}}
\pegrule{Seq}
\qquad
\frac{\pegsem{G}{e_1}{w}{\pegfail}}
     {\pegsem{G}{e_1e_2}{w}{\pegfail}}
\pegrule{Seq-F}
\\[1ex]
\frac{\pegsem{G}{e_1}{w}{y}\quad y\in\symset^*}
     {\pegsem{G}{e_1/e_2}{w}{y}}
\pegrule{Choice-L}
\qquad
\frac{\pegsem{G}{e_1}{w}{\pegfail}\quad
      \pegsem{G}{e_2}{w}{r}}
     {\pegsem{G}{e_1/e_2}{w}{r}}
\pegrule{Choice-R}
\\[1ex]
\frac{\pegsem{G}{e}{w}{\pegfail}}
     {\pegsem{G}{!e}{w}{w}}
\pegrule{Not-S}
\qquad
\frac{\pegsem{G}{e}{w}{y}\quad y\in\symset^*}
     {\pegsem{G}{!e}{w}{\pegfail}}
\pegrule{Not-F}
\end{gathered}
\]

Whenever it is defined, the semantics is deterministic.
More precisely, if $\pegsem{G}{e}{w}{r}$ and $\pegsem{G}{e}{w}{r'}$, then $r=r'$.
This follows by induction on the sum of the heights of the two derivations.
A derivation need not exist, however, because recursive calls may fail to terminate.

\begin{definition}[total PEG and PEG language]
A PEG $G$ is \emph{total} if, for every $e\in\parsingexpressions(\nonterminalset,\symset)$ and every
$w\in\symset^*$, there exists $r\in\symset^*\cup\{\pegfail\}$ such that $\pegsem{G}{e}{w}{r}$.
The language recognized by $G$ and the class of languages recognized by total PEGs are, respectively,
\[
\begin{aligned}
L(G)
  &=\{\,w\in\symset^*:\pegsem{G}{\pegstart}{w}{\varepsilon}\,\},\\
\mathsf{PEG}
  &=\{\,L(G):G\text{ is a total PEG}\,\}.
\end{aligned}
\]
\end{definition}

\begin{remark}[acceptance conventions]
\label{rem:peg-acceptance-conventions}
Ford defines the language of a PEG using success of the start expression, without requiring it to consume the entire
input~\cite{Ford2004PEG}.
In our notation, his convention is
\[
L_{\mathrm{succ}}(G)
=
\{\,w\in\symset^*
\mid
\exists r\in\symset^*\;
\pegsem{G}{\pegstart}{w}{r}
\,\}.
\]
Thus a grammar whose start rule is $S\leftarrow\varepsilon$ recognizes $\symset^*$ under Ford's convention, whereas
it recognizes only $\{\varepsilon\}$ under our definition of $L(G)$.

We use the full-consumption convention because membership throughout this paper concerns the entire input word and
because it agrees with the whole-input acceptance convention of scaffolding automata.
At the level of language classes, however, this distinction does not change the class recognized by total PEGs.
Indeed, let $\mathsf{All}\leftarrow .\,\mathsf{All}/\varepsilon$, which consumes the entire remaining input.
Given a total grammar with start nonterminal $S$, its language under Ford's convention is recognized under our
convention after adding the fresh start rule
$S_{\mathrm{full}}\leftarrow \mathbin{\&}S\,\mathsf{All}$.
Conversely, its language under our convention is recognized under Ford's convention after adding the fresh start
rule $S_{\mathrm{succ}}\leftarrow S\,{!.}$.
Both transformations preserve totality.
\end{remark}

\subsection{Scaffolding Automata}

We use the scaffolding automata of Loff, Moreira, and Reis~\cite{LoffMoreiraReis2020}.

\begin{definition}[scaffold]
Let $d\geq1$, and let $\workalphabet$ be a finite working alphabet.
A $(d,\workalphabet)$-scaffold of size $t+1$ is
\[
\scaffold=(V,(E_v)_{v\in V},\ell),
\qquad
V=\{0,\ldots,t\},
\]
where
\[
\ell\colon V\to\workalphabet\cup\{\scamissing\},
\qquad
E_v\colon[d]\to V\cup\{\scamissing\},
\qquad
E_v(i)\in\{0,\ldots,v\}\cup\{\scamissing\}.
\]
The node $t$ is the top of $\scaffold$.
\end{definition}

For $j\geq0$, let $\scapaths{d}{j}$ denote the set of \emph{paths} of length at most $j$, namely sequences of edge
indices drawn from $[d]$.
For a node $v$ and a path $p=i_1\cdots i_m$, let $v\cdot p$ be the endpoint obtained by following edges
$i_1,\ldots,i_m$ from $v$, and let $v\cdot p=\scamissing$ if any required edge is missing.

\begin{definition}[neighbourhood]
The radius-$j$ neighbourhood of $v$ is defined by
\[
\begin{aligned}
\scaview{0}{\scaffold}{v}
   &=\ell(v),\\
\scaview{j+1}{\scaffold}{v}
   &=
   \bigl(\ell(v),
     \scaview{j}{\scaffold}{E_v(0)},\ldots,
     \scaview{j}{\scaffold}{E_v(d-1)}\bigr),
\end{aligned}
\qquad
\scaview{j}{\scaffold}{\scamissing}=\scamissing.
\]
Let $\scaview{k}{d}{\workalphabet}$ denote the finite set of all radius-$k$ neighbourhoods of
$(d,\workalphabet)$-scaffolds.
\end{definition}

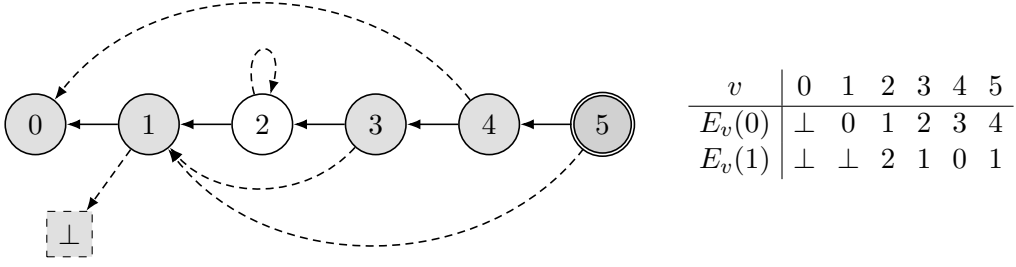
\begin{figure}[h!]
\centering
\begin{tikzpicture}[
  x=1cm,
  y=1cm,
  scaffold node/.style={
    circle,
    draw=black,
    semithick,
    fill=white,
    minimum size=8mm,
    inner sep=1pt
  },
  neighbourhood node/.style={
    scaffold node,
    fill=black!12
  },
  top scaffold node/.style={
    neighbourhood node,
    double,
    fill=black!18
  },
  missing node/.style={
    rectangle,
    draw=black,
    dashed,
    fill=black!12,
    minimum size=6mm,
    inner sep=1pt
  },
  zero edge/.style={
    -{Latex[length=1.8mm]},
    semithick
  },
  one edge/.style={
    -{Latex[length=1.8mm]},
    semithick,
    densely dashed
  }
]
  \node[neighbourhood node] (s0) at (0,1.2) {$0$};
  \node[neighbourhood node] (s1) at (1.5,1.2) {$1$};
  \node[scaffold node] (s2) at (3,1.2) {$2$};
  \node[neighbourhood node] (s3) at (4.5,1.2) {$3$};
  \node[neighbourhood node] (s4) at (6,1.2) {$4$};
  \node[top scaffold node] (s5) at (7.5,1.2) {$5$};
  \node[missing node] (sbot) at (0.45,-0.25) {$\scamissing$};

  \draw[zero edge] (s1) -- (s0);
  \draw[one edge] (s1) -- (sbot);
  \draw[zero edge] (s2) -- (s1);
  \draw[one edge] (s2) to[loop above,min distance=8mm] (s2);
  \draw[zero edge] (s3) -- (s2);
  \draw[one edge] (s3) to[bend left=48] (s1);
  \draw[zero edge] (s4) -- (s3);
  \draw[one edge] (s4) to[bend right=52] (s0);
  \draw[zero edge] (s5) -- (s4);
  \draw[one edge] (s5) to[bend left=52] (s1);

  \node[anchor=west,inner sep=0pt] at (8.5,1.2) {
    \setlength{\arraycolsep}{4pt}
    $\begin{array}{c|cccccc}
      v      &0&1&2&3&4&5\\ \hline
      E_v(0) &\scamissing&0&1&2&3&4\\
      E_v(1) &\scamissing&\scamissing&2&1&0&1
    \end{array}$
  };
\end{tikzpicture}
\caption{A degree-$2$ scaffold with top node $5$.
Solid and dashed arrows denote edges $0$ and $1$, respectively, and shaded vertices belong to the radius-$2$
neighbourhood of the top.}
\label{fig:scaffold-neighbourhood}
\end{figure}

\begin{definition}[scaffolding automaton]
A scaffolding automaton is an $8$-tuple
\[
\automaton=
\langle\symset,d,\workalphabet,k,\stateset,\delta,q_0,F\rangle,
\]
where $\symset$ is the input alphabet, $d\geq1$ is the degree, $\workalphabet$ is the working alphabet, $k\geq0$ is
the view radius, $\stateset$ is a finite state set, $q_0\in\stateset$ is the initial state,
$F\subseteq\stateset$ is the set of accepting states, and
\[
\delta\colon
\stateset\times\symset\times\scaview{k}{d}{\workalphabet}
\longrightarrow
\stateset\times\workalphabet\times
\bigl(\scapaths{d}{k}\cup\{\scaself,\scamissing\}\bigr)^d.
\]
Here $\scaself$ is a transition instruction that makes the corresponding edge of the newly created node point to the
node itself, whereas $\scamissing$ leaves that edge missing.
\end{definition}

The initial scaffold $\scaffold_0$ consists of a single node $0$ whose label and edges are all $\scamissing$.
If the current configuration is $(q,\scaffold)$ with top $t$ and
\[
\delta\bigl(q,\sym,\scaview{k}{\scaffold}{t}\bigr)
=(q',\gamma,p_0,\ldots,p_{d-1}),
\]
then $\scastep_{\delta,\sym}(q,\scaffold)=(q',\scaffold')$, where $\scaffold'$ is obtained by appending a node $t+1$
labelled $\gamma$ and
\[
E_{t+1}(i)=
\begin{cases}
\scamissing,
  &p_i=\scamissing,\\
t+1,
  &p_i=\scaself,\\
\scamissing,
  &p_i\in\scapaths{d}{k}\text{ and }t\cdot p_i=\scamissing,\\
t\cdot p_i,
  &p_i\in\scapaths{d}{k}\text{ and }t\cdot p_i\neq\scamissing.
\end{cases}
\]
For $w=\sym_1\cdots\sym_n$, begin with $(q_0,\scaffold_0)$ and define
\[
(q_{i+1},\scaffold_{i+1})
=\scastep_{\delta,\sym_{i+1}}(q_i,\scaffold_i)
\qquad(0\leq i<n).
\]
The language decided by $\automaton$ and the class of languages decided by scaffolding automata are, respectively,
\[
\begin{aligned}
L(\automaton)
  &=\{\,w\in\symset^*:q_n\in F\,\},\\
\mathsf{SCA}
  &=\{\,L(\automaton):\automaton\text{ is a scaffolding automaton}\,\}.
\end{aligned}
\]

Loff, Moreira, and Reis proved the following characterization~\cite[Theorem~16]{LoffMoreiraReis2020}.

\begin{theorem}[PEG--SCA characterization]
\label{thm:peg-sca-characterization}
For every language $L\subseteq\symset^*$,
\[
L\in\mathsf{PEG}
\quad\Longleftrightarrow\quad
L^R\in\mathsf{SCA}.
\]
\end{theorem}

In particular, for every language $K$,
\[
K\notin\mathsf{SCA}
\quad\Longrightarrow\quad
K^R\notin\mathsf{PEG}.
\]

\begin{remark}[on the proof of \cref{thm:peg-sca-characterization}]
\label{rem:loff-repair}
We use both directions of \cref{thm:peg-sca-characterization}: the necessary
direction in \cref{thm:multiphase-negative}, and the sufficient direction in
\cref{lem:closure-basic,lem:closure-pref-peg}.
The accompanying artifact proves both directions in full and checks them in
Lean.
\end{remark}

\subsection{Cell-Probe Model}

The cell-probe model originates in Yao's work on static information retrieval~\cite{Yao1981}.
Fredman and Saks introduced its dynamic formulation~\cite{FredmanSaks1989}.
We use the following deterministic dynamic formulation.

\begin{definition}[deterministic cell-probe model]
Fix a word size $w\geq1$.
The address space, word space, and memory states are
\[
\celladdr{w}=\cellword{w}=\bitset^w,
\qquad
\cellmemory{w}
=\{\,M:\celladdr{w}\to\cellword{w}\,\}.
\]
A deterministic cell-probe data structure consists of a preprocessing map
\(\mathsf{Pre}\colon\mathcal{I}\to\cellmemory{w}\), where $\mathcal{I}$ is the set of preprocessing inputs,
together with an adaptive procedure for each update and query operation.
Starting from persistent memory $M$, a procedure may repeatedly perform
\[
\mathsf{read}(p): z\gets M(p),
\qquad
\mathsf{write}(p,z): M(p)\gets z,
\qquad
p\in\celladdr{w},\ z\in\cellword{w}.
\]
The cost counts only memory access.
Each read or write costs one probe, whereas preprocessing and all computation between probes are free.
When an operation terminates, its updated memory becomes the persistent memory supplied to the next operation, while
any transient local state is discarded.
\end{definition}

A data structure is \emph{exact} if every query returns the correct answer on every legal operation sequence.
We measure its performance by the worst-case update time $t_u$ and query time $t_q$, namely the maximum numbers of
probes made by a single update and a single query.
Throughout, $n$ denotes the size parameter of the problem under consideration, and we take $w=\Theta(\log n)$, so
that a word can address any of the $2^w=n^{\Theta(1)}$ cells in a memory of polynomial size.

\section{The SCA--Cell-Probe Transfer Framework}
\label{sec:framework}

In this section, we formalize the transfer from SCA recognizers to cell-probe
data structures.  The key observation is that a transition of a fixed SCA
accesses only a constant-size neighbourhood of its scaffold and can therefore
be simulated with $\mathcal{O}(1)$ cell probes.  We first define three-stage
Boolean data-structure problems and their SCA encodings, then prove the
transfer theorem and derive its consequence for PEG lower bounds.

\subsection{Three-Stage Problems and SCA Encodings}
\label{sec:framework-encoding}

The three stages are serialized in their natural order, with the query last.
For an SCA, this ordering produces a query-independent scaffold snapshot.  By
the time the query begins, the SCA has processed the preprocessing data and
updates but has not seen which query it must answer, so the snapshot can serve
as persistent memory.

As a running example, consider a bit vector $x\in\{0,1\}^n$ that is given
during preprocessing, updates that each toggle one coordinate of $x$, and a
final query that asks for the current value of one coordinate.  Writing a
subscripted bracket for a suitably delimited binary encoding of the indicated
stage, an operation sequence can be serialized schematically as
\[
\underbrace{\langle n,x\rangle_{\mathrm{pre}}}_{\text{preprocessing prefix}}\,
\underbrace{\langle i_1\rangle_{\mathrm{upd}}\cdots
\langle i_\ell\rangle_{\mathrm{upd}}}_{\text{update blocks}}\,
\underbrace{\langle q\rangle_{\mathrm{qry}}}_{\text{query suffix}}.
\]
The corresponding language consists of the serializations for which
the answer to the final query is $1$.  This is the form of problem and
encoding captured by the following definitions;
\cref{ex:bit-vector-encoding} revisits the example with concrete encoders.

\begin{definition}[three-stage Boolean data-structure problem]
\label{def:three-stage-problem}
A \emph{three-stage Boolean data-structure problem} $\mathcal{P}$ is a family
indexed by $n\geq1$.
For each $n$, it consists of sets $\mathcal{D}_n$, $\mathcal{U}_n$, and
$\mathcal{Q}_n$, a number of updates $\ell(n)\geq1$, and an answer function
\[
f_n\colon
\mathcal{D}_n\times\mathcal{U}_n^{\ell(n)}\times\mathcal{Q}_n
\longrightarrow \{0,1\}.
\]
\end{definition}

For parameter $n$, a data structure for $\mathcal{P}$ first preprocesses
$d\in\mathcal{D}_n$, then receives updates
$u_1,\ldots,u_{\ell(n)}\in\mathcal{U}_n$ in order, and finally receives a
query $q\in\mathcal{Q}_n$ and returns
$f_n(d,u_1,\ldots,u_{\ell(n)},q)$.
The index $j$ is known to the data structure when the update $u_j$ is
presented; correspondingly, the encoders introduced in the next definition
are indexed by $j$ as well.

\begin{definition}[$(a,b)$-$\mathsf{SCA}$ encoding]
\label{def:sca-encoding}
Let $a,b\colon\mathbb{N}_{>0}\to\mathbb{N}_{>0}$, and let
$K\subseteq\Sigma^*$ be a language over a fixed finite alphabet $\Sigma$.
We call $K$ an \emph{$(a,b)$-$\mathsf{SCA}$ encoding} of $\mathcal{P}$ if there
are constants $c_0,c_1>0$, independent of $n$, and, for every $n\geq1$, maps
\[
\begin{aligned}
\mathsf{Pre}_n&\colon\mathcal{D}_n\to\Sigma^*,\\
\mathsf{Upd}_{n,j}&\colon\mathcal{U}_n\to\Sigma^*
&& (1\leq j\leq\ell(n)),\\
\mathsf{Qry}_n&\colon\mathcal{Q}_n\to\Sigma^*
\end{aligned}
\]
satisfying the following conditions.
\begin{enumerate}
\item[(E1)] \textbf{Correctness.}\par
For every $d\in\mathcal{D}_n$,
$u_1,\ldots,u_{\ell(n)}\in\mathcal{U}_n$, and $q\in\mathcal{Q}_n$,
\[
\begin{aligned}
&\mathsf{Pre}_n(d)\,
\mathsf{Upd}_{n,1}(u_1)\cdots
\mathsf{Upd}_{n,\ell(n)}(u_{\ell(n)})\,
\mathsf{Qry}_n(q)\in K\\
&\hspace{4cm}\Longleftrightarrow
f_n(d,u_1,\ldots,u_{\ell(n)},q)=1.
\end{aligned}
\]

\item[(E2)] \textbf{Update and query lengths.}\par
For every $1\leq j\leq\ell(n)$, $u\in\mathcal{U}_n$, and
$q\in\mathcal{Q}_n$,
\[
|\mathsf{Upd}_{n,j}(u)|\leq a(n),
\qquad
|\mathsf{Qry}_n(q)|\leq b(n).
\]

\item[(E3)] \textbf{Polynomial total length.}\par
For the constants $c_0,c_1$ fixed above and every
$d,u_1,\ldots,u_{\ell(n)},q$ as above,
\[
\left|
\mathsf{Pre}_n(d)\,
\mathsf{Upd}_{n,1}(u_1)\cdots
\mathsf{Upd}_{n,\ell(n)}(u_{\ell(n)})\,
\mathsf{Qry}_n(q)
\right|
\leq c_0n^{c_1}.
\]
\end{enumerate}
\end{definition}

\begin{example}[bit-vector toggling]
\label{ex:bit-vector-encoding}
For each $n\geq1$, let
\[
\mathcal{D}_n=\{0,1\}^n,
\qquad
\mathcal{U}_n=\mathcal{Q}_n=[n],
\qquad
\ell(n)=n.
\]
For the bit-vector problem introduced above, the answer function is
\[
f_n(x,i_1,\ldots,i_n,q)
=
x_q\mathbin{\oplus}
\bigoplus_{j=1}^{n}\mathbf{1}[i_j=q].
\]

Let $r(n)=\lceil\log_2(n+1)\rceil$, and let
$\mathsf{code}_n\colon[n]\to\{0,1\}^{r(n)}$ be the fixed-length binary
encoding.  Let $\mathsf{bin}(n)$ be the binary representation of $n$ without
leading zeroes.  Over the fixed alphabet
$\Sigma=\{0,1,\mathtt{p},\mathtt{u},\mathtt{q},\#\}$, define
\[
\begin{aligned}
\mathsf{Pre}_n(x)
  &=\mathtt{p}\,\mathsf{bin}(n)\#x\#,\\
\mathsf{Upd}_{n,j}(i)
  &=\mathtt{u}\,\mathsf{code}_n(i)
  &&(1\leq j\leq n),\\
\mathsf{Qry}_n(q)
  &=\mathtt{q}\,\mathsf{code}_n(q).
\end{aligned}
\]
Let $K_{\mathrm{bit}}$ consist of all words
\[
\mathsf{Pre}_n(x)\,
\mathsf{Upd}_{n,1}(i_1)\cdots
\mathsf{Upd}_{n,n}(i_n)\,
\mathsf{Qry}_n(q)
\]
for which $f_n(x,i_1,\ldots,i_n,q)=1$.  The language $K_{\mathrm{bit}}$ is an
$(a,b)$-$\mathsf{SCA}$ encoding of this bit-vector problem with
$a(n)=b(n)=r(n)+1=\mathcal{O}(\log(n+1))$.
\end{example}

The encoder types impose locality only on the serialization of each operation.
The $j$th update block is determined by $n$, $j$, and $u_j$, while its effect
may depend on the entire SCA configuration produced by the preceding updates.
No separate length bound is imposed on the preprocessing prefix because
preprocessing is free in the cell-probe model.  The polynomial bound on the
complete encoding ensures that an SCA run on it creates only
$n^{\mathcal{O}(1)}$ nodes, whose addresses fit in
$\Theta(\log n)$-bit words.

\subsection{The Transfer Theorem}
\label{sec:framework-transfer}

We begin with the local simulation property that makes the transfer possible.
The constants hidden below may depend on the fixed SCA, but not on the problem parameter $n$.

\begin{lemma}[constant-probe SCA simulation]
\label{lem:constant-probe-sca-simulation}
\normalfont
Fix an SCA $\automaton$.
For every $M\geq1$, each configuration of $\automaton$ containing at most $M$ scaffold nodes can be represented in
$\mathcal{O}(M)$ cells of word size
\[
w=\lceil\log_2(M+1)\rceil+\mathcal{O}(1).
\]
One input-symbol transition can be simulated and stored persistently using $\mathcal{O}(1)$ probes, provided that
the resulting configuration also contains at most $M$ nodes.
Starting from a stored configuration, a string $z$ whose simulation keeps the total number of nodes at most $M$ can
also be processed using $\mathcal{O}(|z|+1)$ probes without modifying persistent memory.
\end{lemma}

\begin{figure}[ht]
\centering
\begin{tikzpicture}[
  x=1cm,
  y=1cm,
  font=\small,
  memory region/.style={
    draw=black,
    rounded corners,
    fill=black!3
  },
  overlay region/.style={
    memory region,
    densely dashed
  },
  record/.style={
    rectangle,
    draw=black,
    minimum height=6.5mm,
    inner sep=1pt
  },
  label field/.style={
    record,
    minimum width=9mm
  },
  pointer field/.style={
    record,
    minimum width=8mm
  },
  dots field/.style={
    record,
    minimum width=4.5mm
  },
  scaffold node/.style={
    circle,
    draw=black,
    minimum size=7mm,
    inner sep=1pt
  },
  local node/.style={
    scaffold node,
    densely dashed,
    fill=white
  },
  appended node/.style={
    scaffold node,
    fill=black!12
  },
  pointer/.style={
    -{Latex[length=1.7mm]},
    semithick
  },
  secondary pointer/.style={
    pointer,
    densely dashed
  }
]
  \node[font=\small\bfseries] at (7.65,5.75) {(a) Persistent extension};
  \draw[memory region] (0.15,2.05) rectangle (15.15,3.55);

  \begin{scope}[xshift=3.5mm]
  \node[scaffold node] (s0) at (1.9,4.45) {$0$};
  \node at (3.7,4.45) {$\cdots$};
  \node[scaffold node] (su) at (5.5,4.45) {$u$};
  \node at (7.3,4.45) {$\cdots$};
  \node[scaffold node,very thick] (st) at (9.1,4.45) {$t$};
  \node at (10.9,4.45) {$\cdots$};
  \node[appended node] (sn) at (12.7,4.45) {$t+1$};
  \draw[pointer] (su) to[bend left=24] (s0);
  \draw[pointer] (st) to[bend left=24] (su);
  \draw[pointer] (sn) to[bend left=24] (st);
  \draw[secondary pointer] (st) to[bend right=22] (s0);
  \draw[secondary pointer] (sn) to[bend right=22] (su);

  \foreach \x/\v in {1.9/0,5.5/u,9.1/t} {
    \node[label field,font=\scriptsize] at (\x-1.125,3.05) {label};
    \node[pointer field,font=\scriptsize] at (\x-0.275,3.05) {ptr$_0$};
    \node[dots field,font=\scriptsize] at (\x+0.35,3.05) {$\cdots$};
    \node[pointer field,minimum width=10mm,font=\scriptsize] at (\x+1.075,3.05) {ptr$_{d-1}$};
    \node[font=\footnotesize] at (\x,2.38) {$\operatorname{Rec}(\v)$};
  }

  \node[label field,fill=black!12,font=\scriptsize] at (11.575,3.05) {label};
  \node[pointer field,fill=black!12,font=\scriptsize] at (12.425,3.05) {ptr$_0$};
  \node[dots field,fill=black!12,font=\scriptsize] at (13.05,3.05) {$\cdots$};
  \node[pointer field,minimum width=10mm,fill=black!12,font=\scriptsize] at (13.775,3.05)
    {ptr$_{d-1}$};
  \node[font=\footnotesize] at (12.7,2.38) {$\operatorname{Rec}(t+1)$};

  \draw[black!45] (s0.south) -- (1.9,3.43);
  \draw[black!45] (su.south) -- (5.5,3.43);
  \draw[black!45] (st.south) -- (9.1,3.43);
  \draw[black!45] (sn.south) -- (12.7,3.43);
  \end{scope}

  \node[font=\small\bfseries] at (7.65,1.4) {(b) Read-only extension};
  \draw[memory region] (0.25,-1.0) rectangle (8.75,1.0);
  \draw[overlay region] (9.05,-1.0) rectangle (15.05,1.0);
  \node[anchor=north west,font=\footnotesize] at (0.4,0.85) {persistent scaffold};
  \node[anchor=north west,font=\footnotesize] at (9.2,0.85) {query-local overlay};

  \node[scaffold node] (p0) at (1.2,-0.25) {$0$};
  \node at (2.35,-0.25) {$\cdots$};
  \node[scaffold node] (pu) at (3.5,-0.25) {$u$};
  \node at (4.65,-0.25) {$\cdots$};
  \node[scaffold node,very thick] (pt) at (6.0,-0.25) {$t$};
  \node[local node] (o1) at (9.9,-0.25) {$t+1$};
  \node[local node] (o2) at (11.7,-0.25) {$t+2$};
  \node at (13.0,-0.25) {$\cdots$};
  \node[local node,very thick] (or) at (14.15,-0.25) {$t+r$};
  \draw[pointer] (pu) to[bend left=32] (p0);
  \draw[pointer] (pt) to[bend left=32] (pu);
  \draw[pointer] (o1) to[bend right=22] (pt);
  \draw[pointer] (o2) to[bend left=32] (o1);
  \draw[secondary pointer] (pt) to[bend right=18] (p0);
  \draw[secondary pointer] (o2) to[bend right=18] (pt);
\end{tikzpicture}
\caption{Persistent and read-only extensions of a stored scaffold.
Solid and dashed arrows denote pointer slots $0$ and $1$, respectively.
In (a), each vertical line associates a node with its stored record.
Query-created nodes in (b) remain in a local overlay.}
\label{fig:sca-cell-probe-simulation}
\end{figure}
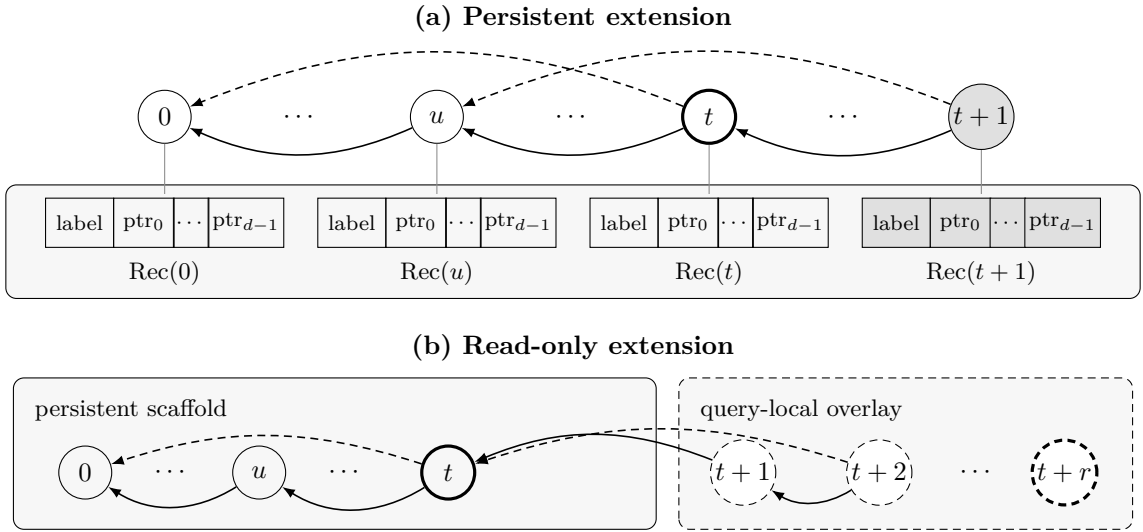

\begin{proof}
Assign each scaffold node $v$ a record containing its label and the $d$ values
$E_v(0),\ldots,E_v(d-1)$.
The current state and top-node identifier are stored in a constant-size header.
Since the degree $d$, the working alphabet, and the state set of $\automaton$ are fixed, this representation uses
$\mathcal{O}(1)$ cells per node and $\mathcal{O}(M)$ cells in total.
Their addresses and stored node identifiers fit in words of the stated size, as illustrated in
\cref{fig:sca-cell-probe-simulation}.

To simulate a transition, recover the radius-$k$ neighbourhood of the current top by following the stored pointers.
At most $1+d+\cdots+d^k$ pointer occurrences are encountered.
Since $d$ and $k$ are fixed, the neighbourhood can be recovered with $\mathcal{O}(1)$ probes, after which the transition
function is evaluated for free and the new record is written with $\mathcal{O}(1)$ further probes.
The resulting probe sequence is summarized in \cref{fig:sca-persistent-transition}.

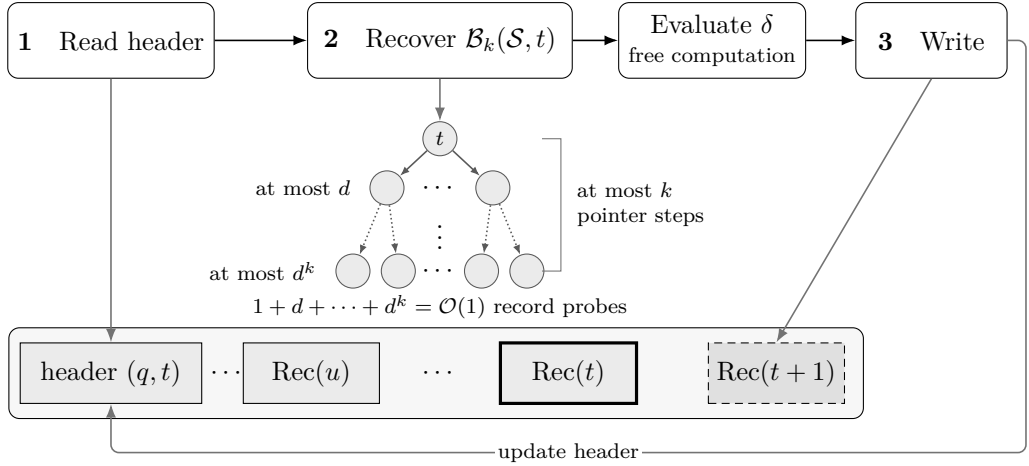
\begin{figure}[ht]
\centering
\begin{tikzpicture}[
  x=1cm,
  y=1cm,
  font=\small,
  memory region/.style={
    draw=black,
    rounded corners,
    fill=black!3
  },
  cell/.style={
    rectangle,
    draw=black,
    minimum width=18mm,
    minimum height=8mm,
    inner sep=2pt
  },
  probed cell/.style={
    cell,
    fill=black!8
  },
  new cell/.style={
    cell,
    densely dashed,
    fill=black!12
  },
  process/.style={
    rectangle,
    draw=black,
    rounded corners,
    minimum height=10mm,
    align=center,
    fill=white
  },
  flow/.style={
    -{Latex[length=1.8mm]},
    semithick
  },
  association/.style={
    flow,
    black!55
  },
  explored occurrence/.style={
    circle,
    draw=black!65,
    fill=black!8,
    minimum size=4.5mm,
    inner sep=0.4pt,
    font=\scriptsize
  },
  exploration edge/.style={
    -{Latex[length=1.3mm]},
    semithick,
    black!65
  }
]
  \draw[memory region] (0,-0.6) rectangle (11.3,0.6);
  \node[probed cell,minimum width=24mm] (header) at (1.35,0) {header $(q,t)$};
  \node at (2.9,0) {$\cdots$};
  \node[probed cell] (recu) at (4.0,0) {$\operatorname{Rec}(u)$};
  \node at (5.7,0) {$\cdots$};
  \node[probed cell,very thick] (rect) at (7.4,0) {$\operatorname{Rec}(t)$};
  \node[new cell] (recn) at (10.15,0) {$\operatorname{Rec}(t+1)$};

  \node[process,minimum width=24mm] (readheader) at (1.35,4.4) {\textbf{1}\quad Read header};
  \node[process,minimum width=35mm] (readview) at (5.7,4.4)
    {\textbf{2}\quad Recover $\scaview{k}{\scaffold}{t}$};
  \node[process,minimum width=23mm] (evaluate) at (9.3,4.4)
    {Evaluate $\delta$\\[-1pt]\scriptsize free computation};
  \node[process,minimum width=20mm] (write) at (12.2,4.4) {\textbf{3}\quad Write};

  \draw[flow] (readheader) -- (readview);
  \draw[flow] (readview) -- (evaluate);
  \draw[flow] (evaluate) -- (write);

  \draw[association] (readheader.south) -- (header.north);
  \node[explored occurrence] (rootocc) at (5.7,3.1) {$t$};
  \node[explored occurrence] (firstleft) at (5.0,2.45) {};
  \node at (5.7,2.45) {$\cdots$};
  \node[explored occurrence] (firstright) at (6.4,2.45) {};
  \node[explored occurrence] (lastone) at (4.55,1.35) {};
  \node[explored occurrence] (lasttwo) at (5.15,1.35) {};
  \node at (5.7,1.35) {$\cdots$};
  \node[explored occurrence] (lastthree) at (6.25,1.35) {};
  \node[explored occurrence] (lastfour) at (6.85,1.35) {};
  \node at (5.7,1.9) {$\vdots$};
  \node[anchor=east,font=\scriptsize] at (4.65,2.45) {at most $d$};
  \node[anchor=east,font=\scriptsize] at (4.2,1.35) {at most $d^k$};
  \draw[exploration edge] (rootocc) -- (firstleft);
  \draw[exploration edge] (rootocc) -- (firstright);
  \draw[exploration edge,densely dotted] (firstleft) -- (lastone);
  \draw[exploration edge,densely dotted] (firstleft) -- (lasttwo);
  \draw[exploration edge,densely dotted] (firstright) -- (lastthree);
  \draw[exploration edge,densely dotted] (firstright) -- (lastfour);
  \draw[black!55] (7.05,3.1) -- (7.3,3.1) -- (7.3,1.35) -- (7.05,1.35);
  \node[anchor=west,align=left,font=\scriptsize] at (7.38,2.23) {at most $k$\\pointer steps};
  \node[font=\scriptsize] at (5.7,0.88)
    {$1+d+\cdots+d^k=\mathcal{O}(1)$ record probes};
  \draw[association] (readview.south) -- (rootocc.north);
  \draw[association] (write.south) -- (recn.north);
  \draw[association,rounded corners=3pt]
    (write.east) -- (13.45,4.4) -- (13.45,-1.05) -- (1.35,-1.05) -- (header.south);
  \node[font=\scriptsize,fill=white,inner sep=1pt] at (7.4,-1.05) {update header};
\end{tikzpicture}
\caption{A persistent transition and the constant-size exploration used to recover
$\scaview{k}{\scaffold}{t}$.}
\label{fig:sca-persistent-transition}
\end{figure}

For a read-only simulation of $z$, keep the current state, top identifier, and every newly created record in the local
overlay shown in \cref{fig:sca-cell-probe-simulation}.
Each edge is either missing or points to a stored node, an earlier overlay node, or the new node itself.
Overlay records are available without probes, while the record of a stored target is obtained by probing memory, so
each transition still costs $\mathcal{O}(1)$ probes.
After every simulated symbol, the stored scaffold together with the local overlay represents the current
configuration of $\automaton$.
The overlay is discarded when the simulation terminates, so persistent memory is unchanged.
\end{proof}

The simulation lemma provides the remaining ingredient for the transfer theorem.

\begin{theorem}[transfer theorem]
\label{thm:transfer}
Let $K$ be an $(a,b)$-$\mathsf{SCA}$ encoding of a three-stage Boolean data-structure problem, and suppose that
$K \in \mathsf{SCA}$.
Then there exists an exact deterministic cell-probe data structure satisfying
\[
w=\Theta(\log n), \qquad
t_u=\mathcal{O}(a(n)), \qquad
t_q=\mathcal{O}(b(n)), \qquad
S=n^{\mathcal{O}(1)},
\]
where $w$ denotes the word size, $t_u$ and $t_q$ denote the numbers of cell
probes required for an update and a query, respectively, and $S$ denotes the
number of persistent memory cells.
\end{theorem}

\begin{proof}
Fix an SCA $\automaton$ with $L(\automaton)=K$.
Let $c_0,c_1>0$ be the constants from condition \textup{(E3)}, and set
$M(n)=\lceil c_0n^{c_1}\rceil+1$.
Condition \textup{(E3)} bounds every legal run by $M(n)$ scaffold nodes, so
\cref{lem:constant-probe-sca-simulation} represents its configurations in
$n^{\mathcal{O}(1)}$ cells of word size $\Theta(\log n)$.

On input $d\in\mathcal{D}_n$, preprocess by simulating $\automaton$ on $\mathsf{Pre}_n(d)$ and storing the resulting
configuration.
For each update $u_j$, compute $\mathsf{Upd}_{n,j}(u_j)$ for free and extend that configuration persistently.
Induction on $j$ gives the configuration of $\automaton$ on the encoded prefix through the $j$th update, while
\textup{(E2)} and \cref{lem:constant-probe-sca-simulation} bound the update cost by
\[
\mathcal{O}\bigl(|\mathsf{Upd}_{n,j}(u_j)|+1\bigr)
=\mathcal{O}(a(n)).
\]

After all updates, simulate $\mathsf{Qry}_n(q)$ from the stored configuration in the read-only overlay of
\cref{lem:constant-probe-sca-simulation}.
Its cost is
\[
\mathcal{O}\bigl(|\mathsf{Qry}_n(q)|+1\bigr)
=\mathcal{O}(b(n)).
\]
The simulation invariant and \textup{(E1)} show that the returned bit is the required value of $f_n$.
Thus the data structure is exact and has all the claimed bounds.
\end{proof}

\subsection{From Cell-Probe Lower Bounds to PEG Lower Bounds}
\label{sec:framework-consequence}

With the transfer theorem in place, the remaining step is a direct combination
with the PEG--SCA characterization.  We record the resulting lower-bound
principle in a form that applies to both of our data-structure problems.

\begin{corollary}[cell-probe lower bounds imply PEG lower bounds]
\label{cor:cell-probe-to-peg}
Let $K$ be an $(a,b)$-$\mathsf{SCA}$ encoding of a three-stage Boolean
data-structure problem.  Suppose that no exact deterministic cell-probe data
structure for this problem simultaneously has
\[
w=\Theta(\log n),\qquad
t_u=\mathcal{O}(a(n)),\qquad
t_q=\mathcal{O}(b(n)),\qquad
S=n^{\mathcal{O}(1)}.
\]
Then $K^R\notin\mathsf{PEG}$.
\end{corollary}

\begin{proof}
If $K^R\in\mathsf{PEG}$, then the PEG--SCA characterization in
\cref{thm:peg-sca-characterization} gives $K\in\mathsf{SCA}$.
The transfer theorem would then produce an exact deterministic cell-probe data
structure with all four displayed bounds, contradicting the hypothesis.
\end{proof}

The corollary becomes useful once a hard data-structure problem is realized by a language with short update and query
blocks.
In the following section, we construct such an encoding for Multiphase Inner Product, verify its parameters against
the known cell-probe lower bound, and establish the language-theoretic properties needed for our main results.

\section{Separation via Multiphase Inner Product}
\label{sec:multiphase}

In this section, we apply the transfer framework to Multiphase Inner Product
and construct a language $C$ satisfying
\[
C\in\mathsf{LIN}\cap\mathsf{PEG}
\qquad\text{and}\qquad
C^R\in\mathsf{LIN}\setminus\mathsf{PEG}.
\]
The non-membership result follows by combining an SCA encoding of Multiphase
Inner Product with Ko's cell-probe lower bound~\cite{Ko2026CellProbe}, while
explicit grammars establish the required membership results.  Consequently,
$\mathsf{LIN}\not\subseteq\mathsf{PEG}$, and the class of PEG languages is
not closed under reversal.

\subsection{The Original Multiphase Problem}
\label{sec:original-multiphase}

P\u{a}tra\c{s}cu introduced the original Multiphase Problem as a dynamic
form of set disjointness~\cite{Patrascu2010Multiphase}.  Phase I preprocesses
sets $S_1,\ldots,S_k\subseteq[n]$.  Phase II receives another set
$T\subseteq[n]$ and updates the memory constructed in Phase I.  Only in
Phase III is an index $i\in\{1,\ldots,k\}$ revealed, at which point the data structure
must decide whether $S_i\cap T=\varnothing$.  In the original word-RAM
formulation, the three phases are allowed
\[
\mathcal{O}(nk\tau),\qquad
\mathcal{O}(n\tau),\qquad
\mathcal{O}(\tau)
\]
time, respectively, for a cost parameter $\tau$.  The Phase II bound is a total budget for processing $T$.
If this processing is decomposed into $n$ coordinate updates, it gives an
amortized budget of $\mathcal{O}(\tau)$ per coordinate, followed by one query
with cost $\mathcal{O}(\tau)$.
Progress on the original problem includes the lower bounds of Ko and
Weinstein under a restricted query-adaptivity
condition~\cite{KoWeinstein2020}; separately, Larsen, Weinstein, and Yu gave
the first superlogarithmic cell-probe lower bounds for dynamic Boolean
problems~\cite{LarsenWeinsteinYu2020}.

\subsection{Multiphase Inner Product and Its Lower Bound}
\label{sec:multiphase-problem}

Ko retains the same three-phase structure but replaces set disjointness with
Inner Product over $\mathbb{F}_2$, and proves an unconditional lower bound
for the resulting problem~\cite{Ko2026CellProbe}.  We use this instantiation
because its lower bound crosses the query-time threshold required by our
transfer theorem.

Multiphase Inner Product separates the two arguments of an inner product
across three stages.  A collection of vectors is available during
preprocessing, a second vector arrives one coordinate at a time, and only
after these updates is the index of the queried vector revealed.

\begin{definition}[Multiphase Inner Product]
\label{def:multiphase-ip}
Fix a function $N\colon\mathbb{N}_{>0}\to\mathbb{N}_{>0}$.
For each $n\geq1$, let
\[
\mathcal{D}_n=(\{0,1\}^n)^{N(n)},
\qquad
\mathcal{U}_n=\{0,1\},
\qquad
\mathcal{Q}_n=\{1,\ldots,N(n)\},
\qquad
\ell(n)=n.
\]
For preprocessing data
$\boldsymbol{S}=(S_1,\ldots,S_{N(n)})\in\mathcal{D}_n$, updates
$x_1,\ldots,x_n\in\mathcal{U}_n$, and a query $q\in\mathcal{Q}_n$, define
\[
f_n(\boldsymbol{S},x_1,\ldots,x_n,q)
=
\langle S_q,x\rangle_{\mathbb{F}_2}
=
\bigoplus_{j=1}^{n}(S_q)_j x_j,
\qquad
x=(x_1,\ldots,x_n).
\]
The resulting three-stage Boolean data-structure problem is
\emph{Multiphase Inner Product over $\mathbb{F}_2$}.
\end{definition}

Thus the $j$th update supplies the coordinate $x_j$, and the final query asks
for the inner product of the completed vector $x$ with one of the
preprocessed vectors.
An update is therefore a single coordinate rather than the whole vector $x$,
and this distinction matters for the hypothesis of
\cref{thm:ko-multiphase} below: where the Phase II bound of the original
problem is a total budget for processing $T$, the update cost $t_u$ is
charged once per coordinate.  This is also the convention of
\cref{thm:transfer}, whose $t_u$ counts the probes of a single update.
We use the regime $N(n)=n^2$.
This choice satisfies the condition on the number of preprocessed vectors in
Ko's lower bound while keeping an index representable with
$\Theta(\log n)$ bits.

Ko proved the following lower bound, stated here using $N(n)$ for the number
of preprocessed vectors.

\begin{theorem}[Ko~\cite{Ko2026CellProbe}, Theorem 1.1]
\label{thm:ko-multiphase}
For Multiphase Inner Product with
$N(n)=n^{1+\Omega(1)}$ and word size $w=\Theta(\log n)$, there is a
constant $\kappa>0$ such that
\[
t_u\leq\log^\kappa n
\qquad\Longrightarrow\qquad
t_{\mathrm{tot}}
=
\widetilde{\Omega}\left(
\left(\frac{\log n}{\log\log n}\right)^2
\right).
\]
Here $t_{\mathrm{tot}}$ is the total number of probes used after the final
query index is revealed, $t_u$ is the number of probes used to process a
\emph{single} update, and $\widetilde{\Omega}$ hides a
$\operatorname{poly}(\log\log n)$ factor in the denominator.  Explicitly, the
conclusion is that there is a constant $c>0$, independent of $n$, with
\[
t_{\mathrm{tot}}
=
\Omega\left(
\frac{1}{(\log\log n)^{c}}
\left(\frac{\log n}{\log\log n}\right)^{2}
\right).
\]
\end{theorem}

\begin{remark}[what the separation needs from \cref{thm:ko-multiphase}]
\label{rem:omega-log-suffices}
The proof of \cref{cor:no-fast-multiphase} uses only that Multiphase Inner
Product admits \emph{some} $\omega(\log n)$ query lower bound under a
polylogarithmic update bound; the exponent $2$ and the factor hidden by
$\widetilde{\Omega}$ are immaterial.  The threshold is essential, though:
an $\Omega(\log n)$ bound would \emph{not} suffice, since the data structure
of \cref{thm:transfer} attains $t_q=\mathcal{O}(\log n)$; the logarithmic
bounds of P\u{a}tra\c{s}cu and Demaine~\cite{PatrascuDemaine2006} are of
exactly the strength that does not.  Crossing that barrier for dynamic
Boolean problems~\cite{LarsenWeinsteinYu2020} is what the argument
consumes.
\end{remark}

For our purposes, the relevant consequence is the following exclusion.

\begin{corollary}
\label{cor:no-fast-multiphase}
For Multiphase Inner Product with $N(n)=n^2$, no exact deterministic
cell-probe data structure simultaneously has
\[
w=\Theta(\log n),\qquad
t_u=\mathcal{O}(1),\qquad
t_q=\mathcal{O}(\log n),\qquad
S=n^{\mathcal{O}(1)}.
\]
\end{corollary}

\begin{proof}
The number of preprocessed vectors satisfies $N(n)=n^2=n^{1+\Omega(1)}$, and $t_u=\mathcal{O}(1)$ is at most
$\log^\kappa n$ for all sufficiently
large $n$.  Since there is one final query, Ko's $t_{\mathrm{tot}}$ is our
$t_q$.  Let $c$ be the constant supplied by \cref{thm:ko-multiphase}.  The
lower bound is superlogarithmic because, for this fixed $c$,
\[
\frac{1}{\log n}
\cdot
\frac{1}{(\log\log n)^{c}}
\left(\frac{\log n}{\log\log n}\right)^{2}
=
\frac{\log n}{(\log\log n)^{c+2}}
\longrightarrow\infty.
\]
It therefore contradicts $t_q=\mathcal{O}(\log n)$, and the argument is
insensitive to the value of $c$.
Finally, \cref{thm:ko-multiphase} places no restriction on the number of
memory cells, so adding the requirement $S=n^{\mathcal{O}(1)}$ only narrows
the class of data structures being excluded.
\end{proof}

\subsection{The Witness Language}
\label{sec:multiphase-language}

To apply the transfer framework, we realize Multiphase Inner Product as a
$(1,\mathcal{O}(\log n))$-$\mathsf{SCA}$ encoding using the record language below.
The same language also supports the linear CFG and PEG constructions needed
for the positive results.

Let $B=\{0,1\}$ and
$\Sigma=B\cup\{\mathtt{<},\mathtt{:},\mathtt{>},
\mathtt{\$},\mathtt{\#}\}$.
For $k,s\in B^+$, the word
$\mathord{\mathtt{<}}k\mathtt{:}s\mathord{\mathtt{>}}$
is a record with key $k$ and payload $s$.

\begin{definition}[witness languages]
\label{def:multiphase-witness}
The \emph{record-first} language $C\subseteq\Sigma^*$ is
\begin{equation}
\label{eq:multiphase-C}
\begin{aligned}
C=\bigl\{\,
&\mathord{\mathtt{<}}k_1\mathtt{:}s_1\mathord{\mathtt{>}}\cdots
 \mathord{\mathtt{<}}k_m\mathtt{:}s_m\mathord{\mathtt{>}}\,
 \mathtt{\$}\,x\,\mathtt{\#}\,q
\ \bigm|\
m\geq1,\ k_i,s_i,x,q\in B^+,\\
&\exists i\in\{1,\ldots,m\}\
\bigl[
k_i^R=q,\quad
|s_i|=|x|,\quad
\langle s_i^R,x\rangle_{\mathbb{F}_2}=1
\bigr]
\,\bigr\}.
\end{aligned}
\end{equation}
The second witness language is its reversal,
\begin{equation}
\label{eq:multiphase-L}
L=C^R.
\end{equation}
\end{definition}

A word of $C$ lists the preprocessed records first, followed by the update
vector $x$ and the query key $q$.  A record witnesses membership when its
reversed key equals $q$ and its reversed payload has odd inner product with
$x$.  Keys and payloads may here have arbitrary nonempty lengths; the
fixed-width encoding used for the lower bound will be imposed only on the
canonical words defined in \cref{sec:multiphase-separation}.

For example,
\[
\mathord{\mathtt{<}}10\mathtt{:}11\mathord{\mathtt{>}}
\mathord{\mathtt{<}}01\mathtt{:}101\mathord{\mathtt{>}}
\mathtt{\$}110\mathtt{\#}10
\]
belongs to $C$.  Its second record satisfies $(01)^R=10$ and
$\langle(101)^R,110\rangle_{\mathbb{F}_2}=1$.

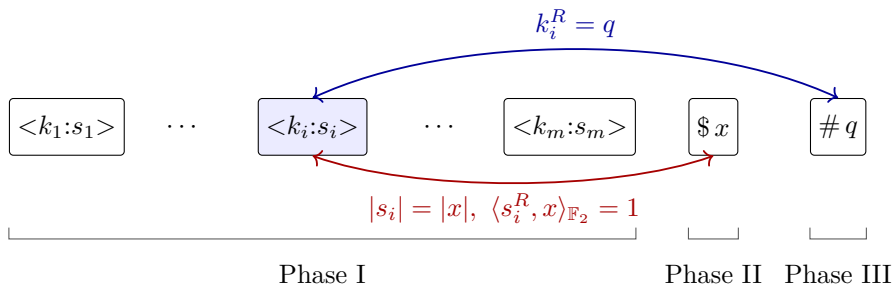
\begin{figure}[ht]
\centering
\begin{tikzpicture}[
  block/.style={
    draw,
    rounded corners=1.5pt,
    minimum height=7.5mm,
    inner sep=3pt,
    font=\small
  },
  label/.style={font=\small}
]
\node[block] (r1) at (0,0)
  {$\mathord{\mathtt{<}}k_1\mathtt{:}s_1\mathord{\mathtt{>}}$};
\node[label] (d1) at (1.55,0) {$\cdots$};
\node[block, fill=blue!8] (ri) at (3.25,0)
  {$\mathord{\mathtt{<}}k_i\mathtt{:}s_i\mathord{\mathtt{>}}$};
\node[label] (d2) at (4.95,0) {$\cdots$};
\node[block] (rm) at (6.65,0)
  {$\mathord{\mathtt{<}}k_m\mathtt{:}s_m\mathord{\mathtt{>}}$};
\node[block] (x) at (8.55,0) {$\mathtt{\$}\,x$};
\node[block] (q) at (10.20,0) {$\mathtt{\#}\,q$};

\draw[<->, blue!60!black, line width=0.7pt]
  (ri.north) to[out=25,in=155,looseness=0.75]
  node[above, label, pos=0.5] {$k_i^R=q$} (q.north);
\draw[<->, red!65!black, line width=0.7pt]
  (ri.south) to[out=-20,in=-160,looseness=0.68]
  node[below, label, pos=0.48]
  {$|s_i|=|x|,\ \langle s_i^R,x\rangle_{\mathbb{F}_2}=1$}
  (x.south);

\draw[black!65]
  ([yshift=-9.5mm]r1.south west) --
  ([yshift=-11mm]r1.south west) --
  ([yshift=-11mm]rm.south east) --
  ([yshift=-9.5mm]rm.south east);
\path
  ([yshift=-11mm]r1.south west) --
  ([yshift=-11mm]rm.south east)
  node[midway, below=2mm, label] {Phase I};
\draw[black!65]
  ([yshift=-9.5mm]x.south west) --
  ([yshift=-11mm]x.south west) --
  ([yshift=-11mm]x.south east) --
  ([yshift=-9.5mm]x.south east);
\path
  ([yshift=-11mm]x.south west) --
  ([yshift=-11mm]x.south east)
  node[midway, below=2mm, label] {Phase II};
\draw[black!65]
  ([yshift=-9.5mm]q.south west) --
  ([yshift=-11mm]q.south west) --
  ([yshift=-11mm]q.south east) --
  ([yshift=-9.5mm]q.south east);
\path
  ([yshift=-11mm]q.south west) --
  ([yshift=-11mm]q.south east)
  node[midway, below=2mm, label] {Phase III};
\end{tikzpicture}
\caption{The three-stage serialization of the record language.}
\label{fig:multiphase-nesting}
\end{figure}

The order displayed in \cref{fig:multiphase-nesting} places the
payload--update comparison inside the key--query comparison.  After choosing
one record, a linear grammar can generate both pairs from outside inward while
retaining only one nonterminal.  We make this argument formal below.

By \cref{eq:multiphase-L}, the words of $L$ are the reversals of the words of
$C$ and therefore have the form
\[
q^R\,\mathtt{\#}\,x^R\,\mathtt{\$}\,
\mathord{\mathtt{>}}s_m^R\,\mathtt{:}\,k_m^R\mathord{\mathtt{<}}
\cdots
\mathord{\mathtt{>}}s_1^R\,\mathtt{:}\,k_1^R\mathord{\mathtt{<}},
\]
in which the query and the update vector now precede all records.

\paragraph{Linear context-freeness.}
Let $G$ be the context-free grammar with start symbol $A$ and nonterminals
$A,D,E,R,D',E',K,P_0,P_1$.  Its productions are
\begin{equation}
\label{eq:multiphase-linear-cfg}
\begin{array}{rcl}
A
&\to&
  \mathord{\mathtt{<}}0D
  \mid \mathord{\mathtt{<}}1D
  \mid \mathord{\mathtt{<}}0K0
  \mid \mathord{\mathtt{<}}1K1,
\\[1mm]
D
&\to&
  0D\mid1D\mid\mathtt{:}0E\mid\mathtt{:}1E,
\\
E
&\to&
  0E\mid1E\mid\mathord{\mathtt{>}}A,
\\[1mm]
R
&\to&
  \mathtt{\$}
  \mid\mathord{\mathtt{<}}0D'
  \mid\mathord{\mathtt{<}}1D',
\\
D'
&\to&
  0D'\mid1D'\mid\mathtt{:}0E'\mid\mathtt{:}1E',
\\
E'
&\to&
  0E'\mid1E'\mid\mathord{\mathtt{>}}R,
\\[1mm]
K
&\to&
  0K0\mid1K1\mid\mathtt{:}P_1\mathtt{\#},
\\[1mm]
P_0
&\to&
  \mathord{\mathtt{>}}R
  \mid0P_00\mid0P_01\mid1P_00\mid1P_11,
\\
P_1
&\to&
  0P_10\mid0P_11\mid1P_10\mid1P_01.
\end{array}
\end{equation}
Every right-hand side contains at most one nonterminal, so $G$ is linear.
The nonterminals $A,D,E$ generate the records preceding a chosen witness,
while $R,D',E'$ generate those following it.  The nonterminal $K$ matches the
witness key with the final query, and $P_0,P_1$ match the witness payload with
the update vector while recording their inner-product parity.

For a nonterminal $N$, write $L_G(N)=\{\,w\in\Sigma^*\mid N\cfgderive{G}^{*}w\,\}$.

To see what remains to be proved, suppose that $A$ chooses the $i$th record
as the witness and that $a$ is the first bit of its key.  The relevant part
of a derivation then has the schematic form
\[
\begin{aligned}
A
&\cfgderive{G}^{*}
u\,\mathord{\mathtt{<}}aKa\\
&\cfgderive{G}^{*}
u\,\mathord{\mathtt{<}}av\mathtt{:}P_1\mathtt{\#}v^Ra\\
&\cfgderive{G}^{*}
u\,\mathord{\mathtt{<}}av\mathtt{:}s
\mathord{\mathtt{>}}d\,x\mathtt{\#}v^Ra .
\end{aligned}
\]
Here $u$ should be the sequence of records before the witness, while $d$
should consist of the records after it followed by $\mathtt{\$}$.  This
schematic form can now be compared directly with
\cref{eq:multiphase-C}.  The key condition is visible since
$v^Ra=(av)^R$, and the remaining payload conditions are isolated in the
derivation from $P_1$.  It must enforce $|s|=|x|$ and
$\langle s^R,x\rangle_{\mathbb{F}_2}=1$.
The next lemma verifies precisely these two nested mechanisms.  Its
$P_p$ invariant tracks the payload length and parity, and its $K$ invariant
wraps the reversed key--query match around that comparison.

\begin{lemma}[grammar invariants]
\label{lem:multiphase-cfg-invariants}
Let $\mathcal{R}=L_G(R)=
(\mathord{\mathtt{<}}B^+\mathtt{:}B^+\mathord{\mathtt{>}})^*\mathtt{\$}$.
For each $p\in\{0,1\}$,
\begin{equation}
\label{eq:multiphase-parity-invariant}
L_G(P_p)=\bigl\{\,s\,\mathord{\mathtt{>}}\,d\,x
\bigm| s,x\in B^*,\ d\in\mathcal{R},\ |s|=|x|,
\langle s^R,x\rangle_{\mathbb{F}_2}=p\,\bigr\}.
\end{equation}
The corresponding key invariant is
\begin{equation}
\label{eq:multiphase-key-invariant}
L_G(K)
=
\{\,v\mathtt{:}y\mathtt{\#}v^R
\mid v\in B^*,\ y\in L_G(P_1)\,\}.
\end{equation}
\end{lemma}

\begin{proof}
The productions for $R,D',E'$ generate the language of record sequences with
nonempty binary fields followed by $\mathtt{\$}$, which gives the displayed
description of $\mathcal{R}$.  The production
$P_0\to\mathord{\mathtt{>}}R$ gives the base case of
\cref{eq:multiphase-parity-invariant} with $s=x=\varepsilon$ and parity zero.
Each recursive production has the form $P_p\to aP_rb$ with $a,b\in B$ and $p=r\oplus ab$.
Indeed, extending $(s,x)$ to $(as,xb)$ changes the parity according to
\[
\langle(as)^R,xb\rangle_{\mathbb{F}_2}
=
\langle s^R,x\rangle_{\mathbb{F}_2}\oplus ab.
\]
The eight productions for $P_0,P_1$ implement these transitions.  Induction
on a derivation proves one inclusion in
\cref{eq:multiphase-parity-invariant}, and induction on $|s|$ proves the
other.

The base production $K\to\mathtt{:}P_1\mathtt{\#}$ gives
\cref{eq:multiphase-key-invariant} for $v=\varepsilon$.  Each production
$K\to aKa$ adds the same bit to the left of $v$ and the right of $v^R$.
Induction on $|v|$ therefore proves the second identity.
\end{proof}

The lemma accounts for every constraint on the selected witness record.  The
remaining nonterminals generate only the records surrounding it, so combining
their roles identifies the full language of $G$.

\begin{proposition}[linear context-freeness]
\label{prop:multiphase-linear}
The grammar $G$ generates $C$.  Consequently,
\[
C\in\mathsf{LIN}
\qquad\text{and}\qquad
L=C^R\in\mathsf{LIN}.
\]
\end{proposition}

\begin{proof}
The first two alternatives for $A$ generate a record with nonempty key and
payload through $D,E$ and then return to $A$.  After generating any number of
such records, one of the other two alternatives chooses a witness record and
produces $\mathord{\mathtt{<}}aKa$ for some $a\in B$.

Write the selected key as $av$.  Substituting
\cref{eq:multiphase-key-invariant,eq:multiphase-parity-invariant} into this
part of the derivation produces
\[
\mathord{\mathtt{<}}\,
\underbrace{av}_{k_i}\,
\mathtt{:}\,
\underbrace{s}_{s_i}\,
\mathord{\mathtt{>}}\,
d\,x\,\mathtt{\#}\,
\underbrace{v^Ra}_{q},
\]
where $d\in\mathcal{R}$ contains the records following the witness and the
delimiter $\mathtt{\$}$.  Hence $q=(av)^R=k_i^R$, $|s_i|=|x|$, and
$\langle s_i^R,x\rangle_{\mathbb{F}_2}=1$.
The selected key is nonempty because of its first bit $a$.  The selected
payload and $x$ are nonempty because a derivation from $P_1$ must use at
least one recursive production before reaching the only base production in
$P_0$.  Thus every word generated by $G$ belongs to $C$.

Conversely, take a word in $C$ and choose a record satisfying the three
conditions in \cref{eq:multiphase-C}.  The productions for $A,D,E$ generate
the preceding records.  Splitting the nonempty witness key as $av$, the
productions for $K$ match it with $q=(av)^R$.  The parity invariant generates
the witness payload, all following records, $\mathtt{\$}$, and $x$.
Therefore $G$ generates every word in $C$, and hence $L(G)=C$.

Finally, reversal preserves linearity.  For any linear grammar, replace each
production $N\to uMv$ by $N\to v^RMu^R$ and each terminal-only production
$N\to u$ by $N\to u^R$.  The resulting linear grammar generates the reversal
of the original language.  Since $C\in\mathsf{LIN}$, it follows that
$L=C^R\in\mathsf{LIN}$.
\end{proof}

By \cref{prop:multiphase-linear}, both witness languages are linear
context-free.  It remains to distinguish them at the level of PEGs.  We first
give a PEG for the record-first language $C$, while the hard encoding in
\cref{sec:multiphase-separation} will establish
$L=C^R\notin\mathsf{PEG}$.

\subsection{A PEG for the Record-First Language}
\label{sec:multiphase-peg}

To prove non-closure under reversal, we construct an explicit total PEG for
$C$, complementing the lower bound $L=C^R\notin\mathsf{PEG}$ proved later.
Unlike the linear grammar above, which can choose a witness record
nondeterministically, this PEG searches the records from left to right and
commits to the first one that satisfies the witness conditions.

\paragraph{The grammar.}
Let $G_C$ be the PEG over $\Sigma$ whose start nonterminal is
$\mathsf{Search}$ and whose rules are
\begin{equation}
\label{eq:multiphase-direct-peg}
\begin{array}{rcl@{\qquad\qquad}rcl}
\mathsf{Bit}
&\leftarrow&
  \mathtt{0}/\mathtt{1}
&
\mathsf{M}_0
&\leftarrow&
  \mathtt{0}\,\mathsf{M}_0\,\mathsf{Bit}
  /\mathtt{1}\,
   (\mathsf{M}_0\,\mathtt{0}/\mathsf{M}_1\,\mathtt{1})
  /\mathord{\mathtt{>}}\,\mathsf{Rest}\,\mathtt{\$}
\\
\mathsf{Bits}
&\leftarrow&
  \mathsf{Bit}\,\mathsf{Bits}/\mathsf{Bit}
&
\mathsf{M}_1
&\leftarrow&
  \mathtt{0}\,\mathsf{M}_1\,\mathsf{Bit}
  /\mathtt{1}\,
   (\mathsf{M}_1\,\mathtt{0}/\mathsf{M}_0\,\mathtt{1})
\\
\mathsf{Entry}
&\leftarrow&
  \mathord{\mathtt{<}}\,\mathsf{Bits}\,\mathtt{:}\,
  \mathsf{Bits}\,\mathord{\mathtt{>}}
&
\mathsf{K}
&\leftarrow&
  \mathtt{0}\,\mathsf{K}\,\mathtt{0}
  /\mathtt{1}\,\mathsf{K}\,\mathtt{1}
  /\mathtt{:}\,\mathsf{M}_1\,\mathtt{\#}
\\
\mathsf{Rest}
&\leftarrow&
  \mathsf{Entry}\,\mathsf{Rest}/\varepsilon
&
\mathsf{Hit}
&\leftarrow&
  \mathord{\mathtt{<}}\,
  (\mathtt{0}\,\mathsf{K}\,\mathtt{0}
   /\mathtt{1}\,\mathsf{K}\,\mathtt{1})\,
  \mathsf{End}
\\
\mathsf{End}
&\leftarrow&
  !.
&
\mathsf{Search}
&\leftarrow&
  \mathsf{Hit}/\mathsf{Entry}\,\mathsf{Search}
\end{array}
\end{equation}
The expression $\mathsf{Search}$ first tries to use the current record as a
witness through $\mathsf{Hit}$.  If that attempt fails, it consumes one record
through $\mathsf{Entry}$ and continues.  Inside $\mathsf{Hit}$, the expression
$\mathsf{K}$ matches the selected key with the final query, while
$\mathsf{M}_p$ matches the selected payload with the update vector and tracks
the required inner-product parity in the subscript $p$.

Since prioritized choice commits to the first successful alternative,
language equations alone do not establish correctness.  We therefore
characterize, for each expression, both the inputs on which it succeeds and
the suffix it leaves unconsumed.

\paragraph{Residual-suffix invariants.}
Let $\mathcal{R}_0=(\mathord{\mathtt{<}}B^+\mathtt{:}B^+\mathord{\mathtt{>}})^*$ be the language of complete record
sequences, so that the language
$\mathcal{R}$ of \cref{lem:multiphase-cfg-invariants} equals
$\mathcal{R}_0\,\mathtt{\$}$.

\begin{lemma}[record scanners]
\label{lem:multiphase-peg-records}
For every $k,s\in B^+$, $d\in\mathcal{R}_0$, and $z\in\Sigma^*$,
\[
\pegsem{G_C}
  {\mathsf{Entry}}
  {\mathord{\mathtt{<}}k\mathtt{:}s\mathord{\mathtt{>}}z}
  {z}
\qquad\text{and}\qquad
\pegsem{G_C}
  {\mathsf{Rest}}
  {d\mathtt{\$}z}
  {\mathtt{\$}z}.
\]
Conversely, every successful evaluation of $\mathsf{Entry}$ consumes
one record with nonempty binary fields.  If $\mathsf{Rest}$ succeeds and
leaves $\mathtt{\$}z$, then it has consumed a word in $\mathcal{R}_0$.
\end{lemma}

\begin{proof}
An induction on the block length shows that $\mathsf{Bits}$ consumes
one nonempty binary block.  The delimiters then give the claims for
$\mathsf{Entry}$, and an induction on the number of records gives those for
$\mathsf{Rest}$.
\end{proof}

The record scanners isolate the unconstrained records following the witness.
We can now characterize the nested payload computation.

\begin{lemma}[parity scanner]
\label{lem:multiphase-peg-parity}
For $p\in B$, $s,x\in B^*$, $d\in\mathcal{R}_0$, and $z\in\Sigma^*$,
\begin{equation}
\label{eq:multiphase-peg-parity}
\begin{aligned}
&\pegsem{G_C}
  {\mathsf{M}_p}
  {s\mathord{\mathtt{>}}d\mathtt{\$}xz}
  {z}
\\
&\hspace{20mm}\Longleftrightarrow\quad
|s|=|x|
\quad\text{and}\quad
\langle s^R,x\rangle_{\mathbb{F}_2}=p.
\end{aligned}
\end{equation}
\end{lemma}

\begin{proof}
We use simultaneous induction on $|s|$.  If $s=\varepsilon$, the first input
symbol is $\mathtt{>}$.  Only $\mathsf{M}_0$ has a matching base alternative,
which consumes $\mathord{\mathtt{>}}d\mathtt{\$}$ by
\cref{lem:multiphase-peg-records}.  It leaves $z$ if and only if
$x=\varepsilon$, in agreement with the parity of the empty inner product.

Suppose $s=as'$ with $a\in B$.  A successful recursive call leaves one final
bit $b$ to be consumed after the call, so $x=x'b$.  The induction hypothesis
identifies the state $r$ of the recursive call with
\(\langle(s')^R,x'\rangle_{\mathbb{F}_2}\), while
\[
\langle(as')^R,x'b\rangle_{\mathbb{F}_2}
=
\langle(s')^R,x'\rangle_{\mathbb{F}_2}\oplus ab.
\]
The alternatives for $\mathsf{M}_0,\mathsf{M}_1$ implement the transitions
$p=r\oplus ab$.  When $a=1$, the two inner alternatives inspect
$\mathtt{0}$ and $\mathtt{1}$ at the same position after the recursive call.
Thus the first alternative cannot succeed at a different position and
prioritized choice preserves the displayed equivalence.
\end{proof}

The key scanner wraps this parity computation in the reversed key--query
comparison.

\begin{lemma}[key scanner]
\label{lem:multiphase-peg-key}
For $v,q,s,x\in B^*$, $d\in\mathcal{R}_0$, and $z\in\Sigma^*$,
\begin{equation}
\label{eq:multiphase-peg-key}
\begin{aligned}
&\pegsem{G_C}
  {\mathsf{K}}
  {v\mathtt{:}s\mathord{\mathtt{>}}d
   \mathtt{\$}x\mathtt{\#}qz}
  {z}
\\
&\hspace{12mm}\Longleftrightarrow\quad
q=v^R,\quad
|s|=|x|,\quad
\langle s^R,x\rangle_{\mathbb{F}_2}=1.
\end{aligned}
\end{equation}
\end{lemma}

\begin{proof}
Induct on $|v|$.  For $v=\varepsilon$, the base alternative invokes
$\mathsf{M}_1$, and \cref{lem:multiphase-peg-parity} followed by
$\mathtt{\#}$ gives the result.  If $v=av'$, the first bit selects the
recursive alternative and its final terminal forces $q=q'a$.  The induction
hypothesis gives $q'=(v')^R$ and preserves the payload conditions, hence
$q=v^R$.  The recursive alternatives begin with different terminals, so
prioritized choice introduces no additional case.
\end{proof}

Adding the outer key bit and the end-of-input test gives the complete witness
test used by $\mathsf{Search}$.

\begin{lemma}[witness test]
\label{lem:multiphase-peg-hit}
For every $y\in\Sigma^*$, the expression $\mathsf{Hit}$ terminates on $y$.  Moreover,
\(\pegsem{G_C}{\mathsf{Hit}}{y}{\varepsilon}\) holds if and only if there exist $k,s,x,q\in B^+$ and
$d\in\mathcal{R}_0$ such that
\[
\begin{gathered}
y=
\mathord{\mathtt{<}}k\mathtt{:}s\mathord{\mathtt{>}}
d\mathtt{\$}x\mathtt{\#}q,\\
k^R=q,\qquad
|s|=|x|,\qquad
\langle s^R,x\rangle_{\mathbb{F}_2}=1.
\end{gathered}
\]
On every other input, $\mathsf{Hit}$ returns $\pegfail$.
\end{lemma}

\begin{proof}
The outer pair of matching bits extends
\cref{lem:multiphase-peg-key} to nonempty $k$ and $q$, while odd parity makes
$s$ and $x$ nonempty.  The final $\mathsf{End}$ requires the residual suffix
to be empty.  Every recursive cycle reached from $\mathsf{Hit}$ consumes a
bit or a complete record, so induction on the remaining input length gives
termination and the failure claim.
\end{proof}

It remains to verify that the complete grammar $G_C$ recognizes $C$ and is total.

\begin{proposition}[correctness]
\label{prop:multiphase-peg-correct}
\[
L(G_C)=C.
\]
\end{proposition}

\begin{proof}
Let $w\in C$ and choose its first witness record.  At every earlier record,
$\mathsf{Hit}$ fails by \cref{lem:multiphase-peg-hit}, after which
$\mathsf{Entry}$ consumes that record and $\mathsf{Search}$ continues.
At the selected record, $\mathsf{Hit}$ consumes the entire remaining input.
Thus $\mathsf{Search}$ succeeds and leaves the empty suffix.

Conversely, suppose $\mathsf{Search}$ consumes an entire input.  Every use of
its recursive alternative consumes one complete record through
$\mathsf{Entry}$.  The final alternative must be a successful
$\mathsf{Hit}$, and \cref{lem:multiphase-peg-hit} shows that its current
record witnesses membership in $C$.  Hence the entire input belongs to $C$.
\end{proof}

\begin{proposition}[totality]
\label{prop:multiphase-peg-total}
The PEG $G_C$ is total.
\end{proposition}

\begin{proof}
Every recursive call of $\mathsf{Bits}$, $\mathsf{M}_0$,
$\mathsf{M}_1$, or $\mathsf{K}$ follows a consumed bit, while every recursive
call of $\mathsf{Rest}$ or $\mathsf{Search}$ follows a consumed record.  Thus
ordering the recursive components by their dependencies and then inducting on
the remaining input length proves termination of every nonterminal call, with
simultaneous induction for $\mathsf{M}_0$ and $\mathsf{M}_1$.  Structural
induction then extends termination to every parsing expression, since each
operator makes only finitely many calls to terminating subexpressions.  Hence
$G_C$ is total.
\end{proof}

\begin{theorem}[membership results]
\label{thm:multiphase-positive}
\[
C\in\mathsf{PEG}
\qquad\text{and}\qquad
L=C^R\in\mathsf{SCA}.
\]
\end{theorem}

\begin{proof}
\Cref{prop:multiphase-peg-correct,prop:multiphase-peg-total} show that $G_C$
is a total PEG recognizing $C$.  The PEG--SCA characterization in
\cref{thm:peg-sca-characterization} then gives $L=C^R\in\mathsf{SCA}$.
\end{proof}

This establishes the required membership results.  It remains to exclude $C$ from
$\mathsf{SCA}$ and $L$ from $\mathsf{PEG}$.

\subsection{The Hard Encoding and Separation}
\label{sec:multiphase-separation}

We now establish the complementary non-membership results by reducing Multiphase Inner Product to membership in
$C$.
The reduction encodes each instance as a word with one-symbol update blocks and a logarithmic-length query suffix,
matching the bounds required by the transfer theorem.
Thus an SCA for $C$ would yield a cell-probe data structure excluded by Ko's lower bound.

\paragraph{The canonical hard words.}
For $n\geq1$, set
\[
N(n)=n^2
\qquad\text{and}\qquad
b(n)=\left\lceil\log_2\bigl(N(n)+1\bigr)\right\rceil.
\]
Let $e_n\colon\{1,\ldots,N(n)\}\longrightarrow B^{b(n)}$ map each index to its zero-padded binary representation of
length $b(n)$.
In particular, $e_n$ is injective.

\begin{definition}[canonical hard word]
\label{def:multiphase-hard-word}
For $\boldsymbol{S}=(S_1,\ldots,S_{N(n)})\in(B^n)^{N(n)}$, $x\in B^n$, and
$q\in\{1,\ldots,N(n)\}$, define
\begin{equation}
\label{eq:multiphase-hard-word}
W_n(\boldsymbol{S},x,q)
=
\left(
\prod_{i=1}^{N(n)}
\mathord{\mathtt{<}}e_n(i)^R\mathtt{:}S_i^R\mathord{\mathtt{>}}
\right)
\mathtt{\$}\,x\,\mathtt{\#}\,e_n(q),
\end{equation}
where the product denotes concatenation in increasing order of $i$.
\end{definition}

The reversals make the fields examined by the witness conditions equal to the original key and payload.

\begin{lemma}[hard-word equivalence]
\label{lem:multiphase-hard-word}
For every $\boldsymbol{S}\in(B^n)^{N(n)}$, $x\in B^n$, and
$q\in\{1,\ldots,N(n)\}$,
\[
W_n(\boldsymbol{S},x,q)\in C
\quad\Longleftrightarrow\quad
\langle S_q,x\rangle_{\mathbb{F}_2}=1.
\]
\end{lemma}

\begin{proof}
All fields of $W_n(\boldsymbol{S},x,q)$ are nonempty, and at the $i$th record the three witness conditions reduce to
$e_n(i)=e_n(q)$, $|S_i|=|x|=n$, and $\langle S_i,x\rangle_{\mathbb{F}_2}=1$.
Injectivity of $e_n$ forces $i=q$, which proves the equivalence.
\end{proof}

\paragraph{The SCA encoding.}
The three stages of the hard word define the required encoding maps.
For $1\leq j\leq n$, let
\begin{equation}
\label{eq:multiphase-encoders}
\begin{aligned}
\mathsf{Pre}_n(\boldsymbol{S})
&=
\left(
\prod_{i=1}^{N(n)}
\mathord{\mathtt{<}}e_n(i)^R\mathtt{:}S_i^R\mathord{\mathtt{>}}
\right)
\mathtt{\$},
\\
\mathsf{Upd}_{n,j}(u)
&=u,
\\
\mathsf{Qry}_n(q)
&=\mathtt{\#}\,e_n(q).
\end{aligned}
\end{equation}
Consequently, for $x=(x_1,\ldots,x_n)$, the complete encoding satisfies
\[
\mathsf{Pre}_n(\boldsymbol{S})\,
\mathsf{Upd}_{n,1}(x_1)\cdots
\mathsf{Upd}_{n,n}(x_n)\,
\mathsf{Qry}_n(q)
=
W_n(\boldsymbol{S},x,q)
\]

\begin{proposition}[encoding parameters]
\label{prop:multiphase-sca-encoding}
The language $C$ is an $(a,\beta)$-$\mathsf{SCA}$ encoding of Multiphase Inner Product with $N(n)=n^2$, where
\[
a(n)=1
\qquad\text{and}\qquad
\beta(n)=b(n)+1=\mathcal{O}(\log n).
\]
\end{proposition}

\begin{proof}
Condition \textup{(E1)} follows from \cref{lem:multiphase-hard-word} and the displayed identity between the complete
encoding and $W_n(\boldsymbol{S},x,q)$.
For condition \textup{(E2)}, every update block has length one, while
$|\mathsf{Qry}_n(q)|=b(n)+1=\mathcal{O}(\log n)$.
Finally, every record has length $n+b(n)+3$, and hence the complete encoding has length
\[
\begin{aligned}
|W_n(\boldsymbol{S},x,q)|
&=
N(n)\bigl(n+b(n)+3\bigr)+n+b(n)+2
\\
&=
\mathcal{O}(n^3).
\end{aligned}
\]
This verifies condition \textup{(E3)}.
\end{proof}

\paragraph{The non-membership results.}
We now combine the encoding proposition with Ko's lower bound.

\begin{theorem}[non-membership]
\label{thm:multiphase-negative}
\[
C\notin\mathsf{SCA}
\qquad\text{and}\qquad
L=C^R\notin\mathsf{PEG}.
\]
\end{theorem}

\begin{proof}
Suppose $C\in\mathsf{SCA}$.
By \cref{prop:multiphase-sca-encoding} the encoding has $a(n)=1$, each update
block being the single symbol $\mathsf{Upd}_{n,j}(u)=u$ of
\eqref{eq:multiphase-encoders}, and $\beta(n)=\mathcal{O}(\log n)$.
\Cref{thm:transfer} would then give an exact deterministic cell-probe data
structure for Multiphase Inner Product with $w=\Theta(\log n)$,
$t_u=\mathcal{O}(1)$, $t_q=\mathcal{O}(\log n)$, and $S=n^{\mathcal{O}(1)}$,
which is the combination excluded by \cref{cor:no-fast-multiphase}.
Thus $C\notin\mathsf{SCA}$, and \cref{thm:peg-sca-characterization} then gives
$L=C^R\notin\mathsf{PEG}$.
\end{proof}

Combining these non-membership results with the linear grammar and the direct PEG from the preceding subsections gives
the complete separation.

\begin{theorem}[main separation]
\label{thm:multiphase-separation}
\[
C\in\mathsf{LIN}\cap\mathsf{PEG}\setminus\mathsf{SCA},
\qquad
L=C^R\in\mathsf{LIN}\cap\mathsf{SCA}\setminus\mathsf{PEG}.
\]
\end{theorem}

\begin{proof}
The claim follows from \cref{prop:multiphase-linear,thm:multiphase-positive,thm:multiphase-negative}.
\end{proof}

\begin{corollary}
\label{cor:multiphase-main-consequences}
\[
\mathsf{LIN}\not\subseteq\mathsf{PEG},
\qquad
\mathsf{CFL}\not\subseteq\mathsf{PEG},
\qquad
\mathsf{PEG}\neq\mathsf{PEG}^R.
\]
\end{corollary}

\begin{proof}
By \cref{thm:multiphase-separation}, $L\in\mathsf{LIN}\setminus\mathsf{PEG}$, and the first two conclusions follow
from $\mathsf{LIN}\subseteq\mathsf{CFL}$.
Moreover, $C\in\mathsf{PEG}$ but $C^R=L\notin\mathsf{PEG}$, proving the third conclusion.
\end{proof}

This confirms the conjecture of Loff, Moreira, and Reis that reversal does not preserve PEG
languages~\cite{LoffMoreiraReis2020}.

\section{Closure Properties of PEG Languages}
\label{sec:closure}

The main separation of the previous section already shows that PEG languages are not closed under reversal, since
its witness satisfies $C\in\mathsf{PEG}$ but $C^R=L\notin\mathsf{PEG}$.
The status of other standard operations was mixed.
Positive closure was known for Boolean operations, inverse homomorphisms, and left concatenation by languages in
the regular closure of the deterministic context-free
languages~\cite{Ford2004PEG,BirmanUllman1973,RubtsovChudinov2024}.
Unrestricted concatenation and Kleene star, however, remained open.
Rubtsov and Chudinov explicitly posed the former as an open question~\cite{RubtsovChudinov2024}, while the later
survey of Lucas, Anantharaman, and Smith recorded both as open and conjectured closure under
concatenation~\cite{LucasAnantharamanSmith2026}.

We prove that $\mathsf{PEG}$ is closed under neither unrestricted concatenation nor Kleene star.
The same witness also gives non-closure under homomorphisms and substitutions.
Writing $\mathsf{REG}$ for the class of regular languages, the failure of concatenation takes the strong form
$\mathsf{PEG}\cdot\mathsf{REG}\nsubseteq\mathsf{PEG}$.
It stands in contrast with the inclusion $\mathsf{REG}\cdot\mathsf{PEG}\subseteq\mathsf{PEG}$ of
\cref{thm:closure-asymmetry}, so concatenation with regular languages exhibits a left--right asymmetry.

The central idea behind all four new non-closure results is to split every word of $L$ immediately after a record
witnessing membership in $L$.
The resulting prefixes form a PEG language $\mathsf{Pref}$, while the remaining records form the regular language
$\mathcal{R}_0^R$, yielding the factorization $L=\mathsf{Pref}\cdot\mathcal{R}_0^R$.
This factorization gives the concatenation counterexample directly, yields the star counterexample through a regular envelope, and
yields the homomorphism and substitution counterexamples through a marked product followed by erasure.

\subsection{The Prefix Language and the Factorization}
\label{sec:closure-factorization}

To establish the common factorization underlying these non-closure results, we now define the prefix language
$\mathsf{Pref}$ and show that appending a regular suffix recovers $L$.
Recall from \cref{def:multiphase-witness} that the words of $L=C^R$ have the form
\[
q^R\,\mathtt{\#}\,x^R\,\mathtt{\$}\,
\mathord{\mathtt{>}}s_m^R\,\mathtt{:}\,k_m^R\mathord{\mathtt{<}}
\cdots
\mathord{\mathtt{>}}s_1^R\,\mathtt{:}\,k_1^R\mathord{\mathtt{<}},
\]
where $m\geq1$, all fields lie in $B^+$, and at least one record satisfies the witness conditions of
\cref{eq:multiphase-C}.
The prefix language keeps the part of such a word up to and including a witness.

\begin{definition}[witness prefix language]
\label{def:closure-pref}
The language $\mathsf{Pref}\subseteq\Sigma^*$ consists of the words
\[
q^R\,\mathtt{\#}\,x^R\,\mathtt{\$}\,
(\mathord{\mathtt{<}}k_m\mathtt{:}s_m\mathord{\mathtt{>}})^R\cdots
(\mathord{\mathtt{<}}k_i\mathtt{:}s_i\mathord{\mathtt{>}})^R,
\qquad
m\geq i\geq1,
\]
with all fields in $B^+$, such that the \emph{last} record satisfies the witness conditions of
\cref{eq:multiphase-C}, namely
\[
k_i^R=q,
\qquad
|s_i|=|x|,
\qquad
\langle s_i^R,x\rangle_{\mathbb{F}_2}=1.
\]
\end{definition}

The phrase ``last record'' is unambiguous.
The delimiters $\mathord{\mathtt{>}},\mathtt{:},\mathord{\mathtt{<}}$ cannot occur inside the binary fields, so a
word determines its decomposition into reversed records uniquely.
We also reuse the language $\mathcal{R}_0=(\mathord{\mathtt{<}}B^+\mathtt{:}B^+\mathord{\mathtt{>}})^*$ of complete
record sequences from \cref{sec:multiphase-peg}; its reversal $\mathcal{R}_0^R$ is regular, because $\mathcal{R}_0$
is given by a regular expression and $\mathsf{REG}$ is closed under reversal.

\begin{lemma}[factorization]
\label{lem:closure-factorization}
$L=\mathsf{Pref}\cdot\mathcal{R}_0^R$.
\end{lemma}

\begin{proof}
For $w\in L$, choose a witnessing record $r_i$.
Cutting $w$ immediately after the reversed encoding of $r_i$ yields a prefix in $\mathsf{Pref}$, while the
remaining reversed records form a word of $\mathcal{R}_0^R$.

Conversely, let $w=uv$ with $u\in\mathsf{Pref}$ and $v\in\mathcal{R}_0^R$.
Appending $v$ only adds complete records and does not change $(x,q)$ or the witnessing record of $u$.
That record therefore remains a witness in $w$, so $w\in L$.
\end{proof}

The regularity of $\mathcal{R}_0^R$ was noted above, so it remains to establish the less immediate claim
$\mathsf{Pref}\in\mathsf{PEG}$.

\begin{lemma}[prefix language is PEG]
\label{lem:closure-pref-peg}
$\mathsf{Pref}\in\mathsf{PEG}$.
\end{lemma}

\begin{proof}
By \cref{thm:peg-sca-characterization} it suffices to give a scaffolding automaton for the reversal
$\mathsf{Pref}^R$, the language of words
\[
\mathord{\mathtt{<}}k_i\mathtt{:}s_i\mathord{\mathtt{>}}
\cdots
\mathord{\mathtt{<}}k_m\mathtt{:}s_m\mathord{\mathtt{>}}\,
\mathtt{\$}\,x\,\mathtt{\#}\,q
\]
whose \emph{first} record satisfies the witness conditions.
Because of the reversal, this first record corresponds to the witness $r_i$, not to the first record $r_1$ of the
underlying $C$-word.

We recognize $\mathsf{Pref}^R$ by a two-tape Turing machine that is \emph{strictly real time}: it
consumes one input symbol per step and never pauses.
Reading the first record, the machine copies its key onto one tape and its payload onto the other;
later records are checked only for well-formedness, which the finite control does.
While reading $x$ after $\mathtt{\$}$, it scans back along the payload tape by one cell per input
symbol, checking $|s_i|=|x|$ and accumulating the parity of
$\langle s_i^R,x\rangle_{\mathbb{F}_2}$ in its control; while reading $q$ after $\mathtt{\#}$ it does
the same on the key tape to check $q=k_i^R$.
Each pass consumes one stored symbol per input symbol, so no step pauses.

A strictly real-time machine with $t$ tapes compiles into a scaffolding automaton of degree $3t$ and
radius $2$, three pointer roles being reserved for each tape and a transition inspecting only a
radius-$2$ neighbourhood of the current top.
For $t=2$ this gives degree $6$ and radius $2$.
Hence $\mathsf{Pref}^R\in\mathsf{SCA}$, and the sufficient direction of
\cref{thm:peg-sca-characterization} yields $\mathsf{Pref}\in\mathsf{PEG}$.
\end{proof}

\subsection{Concatenation, Kleene Star, and Homomorphisms}
\label{sec:closure-concat-star-hom}

The factorization $L=\mathsf{Pref}\cdot\mathcal{R}_0^R$ directly yields a concatenation counterexample and, through
a regular envelope, a Kleene-star counterexample.
Marking the product before erasing a delimiter also gives homomorphism and substitution counterexamples, whereas a
classical theorem gives inverse-homomorphism closure.
Together with known left-regular closure, the concatenation result gives a left--right asymmetry and reflects the
contrast between existential language concatenation and deterministic PEG sequencing.
We begin with the standard facts $\mathsf{REG}\subseteq\mathsf{PEG}$ and Boolean closure.

\begin{lemma}[basic closure properties]
\label{lem:closure-basic}
The class $\mathsf{PEG}$ contains $\mathsf{REG}$ and is closed under union, intersection, and complement.
\end{lemma}

\begin{proof}
Boolean closure follows from Ford's constructions~\cite{Ford2004PEG} together with
\cref{rem:peg-acceptance-conventions}.
If $D\in\mathsf{REG}$, a deterministic finite automaton for $D^R$ is a radius-$0$ scaffolding automaton, so
\cref{thm:peg-sca-characterization} gives $D\in\mathsf{PEG}$.
\end{proof}

\begin{theorem}[non-closure under concatenation]
\label{thm:closure-concat}
$\mathsf{Pref}\in\mathsf{PEG}$ and $\mathcal{R}_0^R\in\mathsf{REG}$, yet
\[
\mathsf{Pref}\cdot\mathcal{R}_0^R=L\notin\mathsf{PEG}.
\]
Consequently $\mathsf{PEG}\cdot\mathsf{REG}\nsubseteq\mathsf{PEG}$, and since
$\mathsf{REG}\subseteq\mathsf{PEG}$ the class $\mathsf{PEG}$ is not closed under concatenation.
\end{theorem}

\begin{proof}
Membership $\mathsf{Pref}\in\mathsf{PEG}$ is \cref{lem:closure-pref-peg}, and $\mathcal{R}_0^R$ is regular as noted
in \cref{sec:closure-factorization}.
By \cref{lem:closure-factorization} the concatenation of these two languages is $L$, which is not in $\mathsf{PEG}$ by
\cref{thm:multiphase-negative}, so $\mathsf{PEG}\cdot\mathsf{REG}\nsubseteq\mathsf{PEG}$.
Since $\mathsf{REG}\subseteq\mathsf{PEG}$ by \cref{lem:closure-basic}, both factors lie in $\mathsf{PEG}$, and the
final claim follows.
\end{proof}

The failure is specific to the order of the factors.
A regular language on the \emph{left} preserves membership.

\begin{theorem}[left--right asymmetry]
\label{thm:closure-asymmetry}
$\mathsf{REG}\cdot\mathsf{PEG}\subseteq\mathsf{PEG}$.
\end{theorem}

\begin{proof}
Let $R$ be regular and $Y\in\mathsf{PEG}$.
By \cref{thm:peg-sca-characterization} it suffices to recognize $(R\cdot Y)^R=Y^R\cdot R^R$ by a
scaffolding automaton; the same characterization gives one, say $\automaton$, for $Y^R$, and we fix a
deterministic finite automaton $D$ for the regular language $R^R$, which reversal has turned into a
suffix that $D$ can read online.
Keep the scaffold dynamics of $\automaton$ and enrich its control by the set of states $D$ reaches on
some suffix of the input read so far whose complementary prefix lies in $Y^R$; there are finitely many
such sets, so degree and radius are unchanged.
Seed the set with the initial state of $D$ when $\varepsilon\in Y^R$, and reinsert that state whenever
the $\automaton$-component accepts, which covers the splits with an empty left and an empty right
factor.
Accepting when the set meets the final states of $D$ is then exactly acceptance of $Y^R\cdot R^R$.
\end{proof}

Rubtsov and Chudinov prove a more general statement, allowing any left factor in the regular closure
of the deterministic context-free languages~\cite{RubtsovChudinov2024}.

The same factorization also yields non-closure under Kleene star by placing the regular tail inside an iterated
language and recovering $L$ through a regular envelope.
Define
\[
\begin{aligned}
U
&=\mathsf{Pref}\cup
\bigl(\mathord{\mathtt{>}}B^+\mathtt{:}B^+\mathord{\mathtt{<}}\bigr),
\\
\mathsf{Shape}
&=B^+\,\mathtt{\#}\,B^+\,\mathtt{\$}\,
\bigl(\mathord{\mathtt{>}}B^+\mathtt{:}B^+\mathord{\mathtt{<}}\bigr)^*.
\end{aligned}
\]
Here $U$ is the union of the prefix language and the regular language of single reversed records, while
$\mathsf{Shape}$ is a regular envelope.

\begin{lemma}[star envelope]
\label{lem:closure-star-shape}
$U^*\cap\mathsf{Shape}=L$.
\end{lemma}

\begin{proof}
Let $w\in U^*\cap\mathsf{Shape}$ and factor it into words of $U$.
There is a unique $\mathsf{Pref}$-factor, since such a factor contains one $\mathtt{\#}$ and one
$\mathtt{\$}$, reversed records contain neither, and $\mathsf{Shape}$ contains each of these symbols once.
This factor comes first because words of $\mathsf{Shape}$ begin with a bit, whereas reversed records begin with
$\mathord{\mathtt{>}}$.
Thus $w\in\mathsf{Pref}\cdot\mathcal{R}_0^R=L$.
Conversely, \cref{lem:closure-factorization} places every word of $L$ in $U^*$, and its defining form places it in
$\mathsf{Shape}$.
\end{proof}

\begin{theorem}[non-closure under Kleene star]
\label{thm:closure-star}
$U\in\mathsf{PEG}$ and $U^*\notin\mathsf{PEG}$.
\end{theorem}

\begin{proof}
Membership $U\in\mathsf{PEG}$ follows from \cref{lem:closure-pref-peg,lem:closure-basic}, since
$U$ is the union of $\mathsf{Pref}$ and a regular language.
Suppose that $U^*\in\mathsf{PEG}$.
Since $\mathsf{Shape}\in\mathsf{REG}\subseteq\mathsf{PEG}$ by \cref{lem:closure-basic}, closure under intersection
in the same lemma would give $U^*\cap\mathsf{Shape}\in\mathsf{PEG}$; by \cref{lem:closure-star-shape}
this language is $L$, contradicting \cref{thm:multiphase-negative}.
\end{proof}

\begin{corollary}[non-closure under homomorphisms]
\label{cor:closure-hom}
$\mathsf{PEG}$ is closed under neither homomorphisms nor substitutions.
\end{corollary}

\begin{proof}
Let $\dagger\notin\symset$ be a fresh symbol and let
$\iota\colon\symset^*\to(\symset\cup\{\dagger\})^*$ be the injective
letter-to-letter homomorphism.
Consider the marked product
\[
\mathsf{Mark}
=
\{\,\iota(u)\dagger\iota(v)
\mid
u\in\mathsf{Pref},\ v\in\mathcal R_0^R
\,\}.
\]
This language lies in $\mathsf{PEG}$.
To see this, take total PEGs for $\mathsf{Pref}$ and $\mathcal R_0^R$ and rename their nonterminals apart.
Lift every terminal of both grammars through $\iota$.
The fresh symbol $\dagger$ cannot be consumed by a lifted terminal, so a lifted expression evaluated immediately
before $\dagger$ has the same behavior as the original expression evaluated at the end of its input.
Combining the lifted grammars under a new start rule that runs the first grammar, consumes $\dagger$, and then runs
the second grammar therefore recognizes $\mathsf{Mark}$.
The lifting and the new start rule preserve totality.

Now define the erasing homomorphism
$e\colon(\symset\cup\{\dagger\})^*\to\symset^*$ by
$e(\iota(\sym))=\sym$ and $e(\dagger)=\varepsilon$.
The image of $\mathsf{Mark}$ under $e$ is
\[
e(\mathsf{Mark})
=
\mathsf{Pref}\cdot\mathcal R_0^R
=
L
\notin
\mathsf{PEG}
\]
by \cref{lem:closure-factorization,thm:multiphase-negative}.
Thus closure under homomorphic images would contradict the marked product's membership in $\mathsf{PEG}$.

Finally, view $e$ as the substitution sending each letter to the singleton language containing its image word.
Every component language is regular and hence belongs to $\mathsf{PEG}$ by \cref{lem:closure-basic}, while the
substitution image of $\mathsf{Mark}$ is again $L$.
Therefore $\mathsf{PEG}$ is not closed under substitutions, even when every component language is itself in
$\mathsf{PEG}$.
\end{proof}

This negative result contrasts with the known positive closure result for inverse homomorphisms.

\begin{theorem}[closure under inverse homomorphisms]
\label{thm:closure-invhom}
If $L\in\mathsf{PEG}$ and $h\colon\Delta^*\to\symset^*$ is a homomorphism, possibly erasing, then
$h^{-1}(L)\in\mathsf{PEG}$.
\end{theorem}

\begin{proof}
Birman and Ullman's gTS formalism, later called the generalized top-down parsing language (GTDPL), is an early
recognition model for top-down parsing with backtracking.
They proved that the languages recognized by this formalism are closed under inverse images under possibly erasing
homomorphisms~\cite{BirmanUllman1973}.
Ford later showed that gTS and PEGs have the same recognition power~\cite{Ford2004PEG}, so the result transfers to
PEG languages.
\end{proof}

Greibach proved that there is a single context-free language $L_0$, a \emph{hardest} context-free language,
such that every context-free language is an inverse homomorphic image of $L_0$ or of
$L_0\setminus\{\varepsilon\}$~\cite{Greibach1973Hardest}.

\begin{corollary}[hardest context-free languages are not PEG languages]
\label{cor:closure-greibach}
$L_0\notin\mathsf{PEG}$.
\end{corollary}

\begin{proof}
If $L_0\in\mathsf{PEG}$, then so is $L_0\setminus\{\varepsilon\}$, being the intersection of $L_0$ with a
regular language (\cref{lem:closure-basic}).
By \cref{thm:closure-invhom} every context-free language would then lie in $\mathsf{PEG}$, contradicting
\cref{cor:multiphase-main-consequences}.
\end{proof}

Loff, Moreira, and Reis observed that the existence of a context-free language without a PEG is equivalent
to $L_0$ having no PEG~\cite{LoffMoreiraReis2020}; \cref{cor:closure-greibach} settles that formulation,
with a fixed language studied since 1973 in place of a witness built for the reduction.

\begin{samepage}
\subsection{Summary of Closure Properties}
\label{sec:closure-summary}

The resulting closure landscape is summarized below.
Shaded rows are established in this work.

\begin{center}
\begin{tabular}{@{}llcl@{}}
\hline
Operation & Form & Closed & Reference \\
\hline
Boolean operations & $L_1\cup L_2,\ L_1\cap L_2,\ \overline{L}$ & Yes & Ford~\cite{Ford2004PEG} \\
\rowcolor{black!6}
Concatenation\footnotemark & $L_1L_2$ & No & This work \\
\rowcolor{black!6}
Kleene star & $L^*$ & No & This work \\
\rowcolor{black!6}
Reversal & $L^R$ & No & This work \\
\rowcolor{black!6}
Homomorphism & $h(L)$ & No & This work \\
Inverse homomorphism & $h^{-1}(L)$ & Yes &
Birman and Ullman~\cite{BirmanUllman1973}, Ford~\cite{Ford2004PEG} \\
\rowcolor{black!6}
Substitution & $\sigma(L)$ & No & This work \\
\hline
\end{tabular}
\end{center}
\end{samepage}
\footnotetext{Concatenation remains asymmetric with regular factors.
\Cref{thm:closure-asymmetry} gives $\mathsf{REG}\cdot\mathsf{PEG}\subseteq\mathsf{PEG}$, a case of a more
general theorem of Rubtsov and Chudinov~\cite{RubtsovChudinov2024}, whereas \cref{thm:closure-concat} gives
$\mathsf{PEG}\cdot\mathsf{REG}\nsubseteq\mathsf{PEG}$.}

\section{Conclusion and Open Problems}
\label{sec:conclusion}

We developed a transfer framework that interprets the scaffold maintained by an SCA as persistent memory in the
cell-probe model.
Since each transition accesses only a fixed-radius neighbourhood, every input symbol can be simulated using
$\mathcal{O}(1)$ probes.
Together with the PEG--SCA characterization, this converts cell-probe lower bounds into expressiveness lower bounds
for PEGs.

Applying the framework to Multiphase Inner Product yields a language $C$ satisfying
\[
C\in\mathsf{LIN}\cap\mathsf{PEG}\setminus\mathsf{SCA},
\qquad
C^R\in\mathsf{LIN}\cap\mathsf{SCA}\setminus\mathsf{PEG}.
\]
Consequently, not every linear context-free language is recognized by a PEG, and the class of PEG languages is not
closed under reversal.

The same witness also gives negative answers to the open closure questions for concatenation and Kleene star.
Together with the positive closure results, this yields the strict asymmetry
\[
\mathsf{REG}\cdot\mathsf{PEG}\subseteq\mathsf{PEG}
\qquad\text{but}\qquad
\mathsf{PEG}\cdot\mathsf{REG}\nsubseteq\mathsf{PEG},
\]
and the known inverse-homomorphism closure follows from Birman and Ullman's theorem and Ford's gTS--PEG
equivalence.

Beyond these results, the proof suggests a broader program for developing PEG lower bounds.
The transfer theorem avoids a direct analysis of PEG computations by moving the lower-bound problem to the
cell-probe model, where strong techniques are already available.
This advantage also marks the present limitation of the method.
It establishes non-membership only after the language has been connected to a data-structure problem with a suitable
lower bound, and therefore does not isolate an intrinsic obstruction within the SCA model.

This limitation motivates an intrinsic theory of SCA lower bounds.
It remains open whether SCA non-membership can be derived directly from combinatorial or information-theoretic
properties of append-only scaffolds and fixed-radius access.
Combined with the PEG--SCA characterization, such a theory could yield reusable PEG non-membership criteria without
requiring a separate data-structure reduction.

An intrinsic theory would do more than provide another proof of the present separation.
By identifying general obstructions rather than individual hard encodings, it could help organize the boundary
between $\mathsf{LIN}$ and $\mathsf{PEG}$.
Our result shows that this boundary is nontrivial, but it does not characterize
$\mathsf{LIN}\cap\mathsf{PEG}$.
Finding such a characterization, even through natural necessary or sufficient conditions for a linear context-free
language to be recognized by a total PEG, remains open.

\clearpage
\section*{Formal Verification}
\label{sec:formal-verification}

The main results of this paper are checked in Lean~4~\cite{MouraUllrich2021Lean4} by the accompanying
artifact~\cite{KimPark2026Artifact}, archived at \url{https://doi.org/10.5281/zenodo.22099762}.
Two entry points cover them, \texttt{CFLNotPEG/Main.lean} for the separation and
\texttt{Closure/Main.lean} for the closure properties; checking either one also prints the axioms
that each of its declarations depends on.
The table pairs each result with a Lean declaration establishing it, naming declarations relative to
the \texttt{PegSeparation} namespace.

\begin{center}
\small
\begingroup
\renewcommand{\arraystretch}{1.14}
\arrayrulecolor{black!45}
\rowcolors{2}{black!3}{white}
\begin{tabular}{
>{\raggedright\arraybackslash}p{0.22\textwidth}
>{\raggedright\arraybackslash}p{0.50\textwidth}
>{\raggedright\arraybackslash}p{0.16\textwidth}
}
\hline
\rowcolor{black!9}
Paper result & Lean declaration & Trust boundary \\
\hline
\Cref{thm:transfer} &
\texttt{Reduction.GenericTransfer.\allowbreak{}genericTransferTheorem} &
Standard \\
\Cref{prop:multiphase-linear,prop:multiphase-peg-correct,prop:multiphase-peg-total} &
\texttt{Witness.LanguageCFG.LinearGrammar.\allowbreak{}cfg\_language\_eq\_C},
\texttt{Witness.DirectPEG.\allowbreak{}C\_recognizedByTotalPEG} &
Standard \\
\Cref{cor:no-fast-multiphase} &
\texttt{CFLNotPEG.\allowbreak{}ko2026MultiphaseIPLowerBound} &
Ko~\cite{Ko2026CellProbe} \\
\Cref{thm:multiphase-separation,cor:multiphase-main-consequences} &
\texttt{CFLNotPEG.finalSeparation},
\texttt{Closure.reversalNonClosure} &
Ko~\cite{Ko2026CellProbe} \\
\Cref{lem:closure-factorization,lem:closure-pref-peg,lem:closure-star-shape} &
\texttt{Closure.witness\_factorization},
\texttt{pref\_isPEG},
\texttt{cFirst\_isSCA\_degree\_six\_radius\_two},
\texttt{starWitness\_kstar\_inter\_shape} &
Standard \\
\Cref{lem:closure-basic,thm:closure-asymmetry} &
\texttt{Closure.recognizedByTotalPEG\_union},
\texttt{\_inter}, \texttt{\_compl}, \texttt{\_of\_isRegular},
\texttt{\_mul\_of\_isRegular\_left} &
Standard \\
\Cref{thm:closure-concat,thm:closure-star,cor:closure-hom} &
\texttt{Closure.concatenationNonClosure},
\texttt{kleeneStarNonClosure},
\texttt{HomomorphismInternal.\allowbreak{}homomorphismNonClosureInternal} &
Ko~\cite{Ko2026CellProbe} \\
\Cref{thm:closure-invhom} &
\texttt{Closure.\allowbreak{}recognizedByTotalPEG\_inverseImage} &
Birman--Ullman and Ford~\cite{BirmanUllman1973,Ford2004PEG} \\
\Cref{cor:closure-greibach} &
\texttt{Closure.GreibachHardest.\allowbreak{}greibachHardest\_isContextFree\_not\_totalPEG} &
All three \\
\hline
\end{tabular}
\endgroup
\end{center}

\emph{Standard} means that only Lean and mathlib's foundational axioms
(\texttt{propext}, \texttt{Classical.choice}, \texttt{Quot.sound}) are reported.
Exactly three further axioms occur above: the project-facing corollary of \cref{thm:ko-multiphase},
which every negative result inherits; inverse-image closure for possibly erasing homomorphisms, which
the artifact postulates in the composite form given by Birman--Ullman together with Ford's
equivalence; and Greibach's hardest-language theorem.
\Cref{cor:closure-greibach} is the only result using all three.
Every positive result in the table is proved outright.

The artifact avoids asymptotic notation, so \cref{thm:transfer} and \cref{cor:no-fast-multiphase}
appear there with explicit budget functions and existentially quantified constants.
The scaffolding-automaton compiler produces exactly that form, so the separation is unaffected.
The remaining rows correspond literally.

\section*{Generative AI Usage}

A generative AI tool based on a large language model was used to examine the witness language and preliminary proof
outlines and to generate all Lean 4 proof scripts in the accompanying formal-verification artifact.
The authors reviewed and revised the manuscript arguments, checked that the formal statements correspond to the
claims reported in the paper, and confirmed that the artifact is accepted by Lean 4.
The authors take full responsibility for every definition, claim, proof, and formalization in the paper.

\bibliographystyle{plain}
\bibliography{references}

\end{document}